\documentclass[aoas,preprint]{imsart}

\usepackage{amsmath,amsfonts,amssymb,graphicx,psfrag}
\usepackage[colorlinks=true,urlcolor=blue,citecolor=black,linkcolor=blue]{hyperref}

\usepackage{dcolumn}
\newcolumntype{.}{D{.}{.}{-1}}
\newcolumntype{d}[1]{D{.}{.}{#1}}
\usepackage{natbib}
\bibpunct{(}{)}{;}{a}{}{,}
\usepackage{comment}
\usepackage{multirow}
\usepackage{lscape}
\usepackage{bbm}
\usepackage{scalefnt}
\usepackage[autolanguage]{numprint}
\usepackage{ctable}
\usepackage{paralist}

\usepackage{booktabs}
\usepackage{bm}
\usepackage{rotating}

\usepackage{caption}
\usepackage{subcaption}

 \usepackage{enumitem}
 \usepackage{dsfont}

\usepackage{amsthm}
\newtheorem{theorem}{Theorem}[section]
\newtheorem{lemma}[theorem]{Lemma}

\newtheorem{algorithm}[theorem]{Algorithm}

\newcommand\indep{\protect\mathpalette{\protect\independenT}{\perp}}

\newcommand{\E}{\mathbb{E}}

\newcommand\independent{\protect\mathpalette{\protect\independenT}{\perp}}
\def\independenT#1#2{\mathrel{\rlap{$#1#2$}\mkern2mu{#1#2}}}

\usepackage{color}

\begin{document}

\begin{frontmatter}
\title{Doubly robust machine learning for an instrumental variable study of surgical care for cholecystitis}
\runtitle{DRML for IV Methods}

  \thankstext{T1}{We thank seminar participants at CMU for helpful comments.}

  \begin{aug}
     \author{\fnms{Kenta} \snm{Takatsu}\thanksref{m2}\ead[label=e1]{ktakatsu@andrew.cmu.edu}},
     \author{\fnms{Alexander W.} \snm{Levis}\thanksref{m2}\ead[label=e2]{alevis@cmu.edu}},
     \author{\fnms{Edward} \snm{Kennedy}\thanksref{m2}\ead[label=e3]{edward@stat.cmu.edu}},
     \author{\fnms{Rachel} \snm{Kelz}\thanksref{m1}\ead[label=e4]{Rachel.Kelz@pennmedicine.upenn.edu}},
     \and
     \author{\fnms{Luke J.} \snm{Keele}\thanksref{m1}\ead[label=e5]{luke.keele@gmail.com}}
     
    \runauthor{K. Takatsu et al.}

    \affiliation{Carnegie Mellon University\thanksmark{m2} and University of Pennsylvania\thanksmark{m1}}
     
     \address{Department of Statistics and Data Science\\
            5000 Forbes Ave, \\
            Pittsburgh, PA 15213 \\
            \printead{e1}}

     \address{Department of Statistics and Data Science\\
            5000 Forbes Ave, \\
            Pittsburgh, PA 15213 \\
            \printead{e2}}
            
     \address{Department of Statistics and Data Science\\
            5000 Forbes Ave, \\
            Pittsburgh, PA 15213 \\
            \printead{e3}}            
              
    \address{Hospital of the University of Pennsylvania\\
            3400 Spruce St. \\
            Philadelphia, PA 19104\\
      \printead{e4}}   

    \address{Hospital of the University of Pennsylvania\\
            3400 Spruce St. \\
            Philadelphia, PA 19104\\
      \printead{e5}}

  \end{aug}

\begin{abstract}
Comparative effectiveness research frequently employs the instrumental variable design since randomized trials can be infeasible for many reasons. In this study, we investigate and compare treatments for emergency \textit{cholecystitis} --- inflammation of the gallbladder. A standard treatment for cholecystitis is surgical removal of the gallbladder, while alternative non-surgical treatments include managed care and pharmaceutical options. As randomized trials are judged to violate the principle of equipoise, we consider an instrument for operative care: the surgeon's tendency to operate. Standard instrumental variable estimation methods, however, often rely on parametric models that are prone to bias from model misspecification. We outline instrumental variable estimation methods based on the doubly robust machine learning framework. These methods enable us to employ various machine learning techniques for nuisance parameter estimation and deliver consistent estimates and fast rates of convergence for valid inference. We use these methods to estimate the primary target causal estimand in an IV design. Additionally, we expand these methods to develop estimators for heterogeneous causal effects, profiling principal strata, and a sensitivity analyses for a key instrumental variable assumption. We conduct a simulation study to demonstrate scenarios where more flexible estimation methods outperform standard methods. Our findings indicate that operative care is generally more effective for cholecystitis patients, although the benefits of surgery can be less pronounced for key patient subgroups.
\end{abstract}


\end{frontmatter}


\section{Introduction}

\subsection{Emergency Treatment for Cholecystitis}

Emergency general surgery (EGS) refers to medical emergencies where the injury is endogenous (e.g., a burst appendix) while trauma care refers to injuries that are exogenous (e.g., a gunshot wound). Recent research has focused on the comparative effectiveness of operative versus non-operative care for various EGS conditions \citep{hutchings2022effectiveness,kaufman2022operative,moler2022local,grieve2023clinical}. One critical finding from this research is that the effectiveness of operative care tends to vary with EGS condition, highlighting the need for additional research for specific EGS conditions. 

One of the most common EGS conditions is \textit{cholecystitis}, an inflammation of the gallbladder. Cholecystitis is often caused by the formation of \textit{gallstones} --- small stones made from cholesterol, bile pigment and calcium salts --- blocking the tube leading out of the gallbladder, but it may also result from bile duct problems, tumors, serious illness and other types of infections. In severe cases, cholecystitis is considered a medical emergency where the standard treatment is the removal of the gallbladder through surgery, a procedure called \textit{cholecystectomy}. Nevertheless, non-operative alternatives including observation, minimally invasive procedures, medication, and supportive care are also available. Given the potential risks associated with surgical interventions for some patients, it is critical to investigate the relative efficacy of surgical interventions compared to non-operative alternatives. Moreover, given the heterogeneity across EGS conditions, it is also necessary to produce evidence focused on specific conditions such as cholecystitis.

This article aims to assess the efficacy of surgery for cholecystitis. One of the main challenges in addressing the effectiveness of surgery is the lack of randomized controlled trials as they are judged to violate the principal of equipoise in this context. Consequently, researchers must rely on observational studies to collect evidence. While observational studies can provide a large sample size, they are subject to confounding bias by indication. This problem arises when the selection of patients to treatment depends on prognostic factors that indicate potential benefits. It is thus essential to adjust for confounding bias to ensure accurate statistical conclusions. Many statistical procedures for observational studies assume all confounding variables, including prognostic factors, are recorded. However, medical claims data may not always contain complete records of all relevant prognostic factors, for instance, physiological measures of patient frailty \citep{keele2018icubeds}. To address these challenges, the literature on EGS has adopted instrumental variable (IV) methods, which can account for unobserved confounding variables, allowing for a more accurate assessment of the effectiveness of surgery for cholecystitis. In our study, adopting an IV design, we focus on new flexible estimation methods.

The current literature on EGS has primarily focused on estimating average treatment effects, which can mask possible patient-to-patient variation \citep{ding2019decompose}. For instance, the average effect of operative care could obscure the fact that some patients respond dramatically well to treatment while others may suffer adverse reactions. One strategy for understanding heterogenous treatment effects is to focus on conditional average treatment effects (CATEs), which describe how treatment effects vary with measured features. CATEs can help design treatment strategies to target patients who are likely to benefit from existing interventions. Despite recent progress in statistical methodologies for flexible CATE estimation \citep{athey2016recursive, nie2017quasi, semenova2017estimation, wager2017estimation, foster2019orthogonal}, current EGS research often relies on single, binary effect modifiers using extant methods \citep{rosen2022analyzing}. Therefore in our study, we also focus on developing methods that allow for more complex combinations of possible effect modifiers. In the following section, we outline the IV design that we employ to study the comparative effectiveness of treatments for cholecystitis.

\subsection{An IV Design for Cholecystitis}
\label{sec:iv_design_for_Cholecystitis}

IV methods refer to a set of study designs and statistical techniques that can facilitate the identification of causal effects in the presence of unobserved confounders. An IV is a variable that is associated with the treatment of interest but affects outcomes only indirectly through its impact on treatment assignment \citep{Angrist:1996}. In the context of our study, an instrument must be associated with treatment via surgical care, but without directly affecting patient outcomes. To be considered a valid instrument, a variable must meet three conditions: (1) it must be associated with the treatment; (2) it must be randomly or as-if randomly assigned; and (3) it cannot have a direct effect on the outcome \citep{Angrist:1996}. When these conditions are met, along with certain structural assumptions about the data-generating process (which we will elaborate on shortly), an IV can provide a consistent estimate of a causal effect even in the presence of unobserved confounding between the treatment and the outcome. See \citet{Baiocchi:2014} and \citet{imbens2014instrumental} for reviews.

Comparative effectiveness research often uses a physician's preference for a specific course of treatment as an IV \citep{brookhart2006evaluating,keeleegsiv2018}. This is based on the assumption that, for a patient with a given level of disease severity, physicians may have different preferences over a course of treatment for idiosyncratic reasons. In other words, the assignment to each physician serves as a ``nudge" toward each mode of care, assuming that the the physician's preference has no direct effect on patient outcomes. \citet{brookhart2006evaluating} proposed this type of IV in the context of drug prescriptions. In our study, we use a surgeon's tendency to operate (TTO) as an instrument to determine whether a patient receives surgery after admission to the emergency department. \citet{keeleegsiv2018} first proposed using TTO as an IV and measured it by calculating the percentage of times a surgeon operated when presented with an EGS condition. The assignment of surgeons to patients is plausibly as-if randomized, since this study focuses on patients who are receiving emergency care and are unlikely to be able to select their physician \citep{keeleegsiv2018}.

The dataset used in this study merges the American Medical Association (AMA) Physician Masterfile with all-payer hospital discharge claims from New York, Florida and Pennsylvania in 2012-2013. The study population includes all patients admitted for emergency or urgent inpatient care. The data includes patient sociodemographic and clinical characteristics, such as indicators for frailty, severe sepsis or septic shock, and 31 comorbidities based on Elixhauser indices \citep{elixhauser1998comorbidity}. Each patient is linked with his or her surgeon through a unique identifier. The primary outcome is the presence of an adverse event after surgery, i.e., death or a prolonged length of stay in the hospital (PLOS). We excluded surgeons from our study who did not perform at least five operations for one of the 51 specific EGS conditions per year within the two-year study timeframe. To measure the instrument (TTO), we follow the strategy proposed by \citet{keeleegsiv2018}. For each surgeon, we randomly split the patient population into five subsets. Using one of the sub-splits, we calculate the proportion of times a surgeon operated for the population of patients for a set of EGS conditions. This measure is used as the instrument for the other 80\% of the patient population. 

\subsection{Our Contributions}
\label{sec:intro_DRML}

The majority of IV analyses use a method known as two-stage least squares (TSLS) for estimation \citep{angrist1995two}. TSLS is based on a set of two linear models that may be prone to bias from model misspecification. Notably, the validity of TSLS relies on correct specification of two parametric linear models related to the conditional mean of the outcome. Consequently, substantial bias can occur if there are nonlinear effects on the outcome not accounted for in the statistical model or interactions between control variables that are omitted. 

To reduce bias from model misspecification, nonparametric methods --- often referred to as machine learning (ML) --- can be used to flexibly estimate parts of the data-generating distribution. However, it is generally infeasible to directly apply ML methods to causal inference, since the resulting estimators may be biased and yield confidence intervals with poor coverage \citep{chernozhukov2018generic}. Specifically, treatment effect estimates based on ML methods may inherit first-order smoothing bias from the ML fits \citep{kennedy2022semiparametric}. In addition, ML methods often exhibit slow convergence rates, which results in inefficiency and invalid statistical inference \citep{van2003unified, robins2008higher,kennedy2016semiparametric, chernozhukov2018double}. Recent research has demonstrated that nonparametric methods can be applied to causal inference with the help of semiparametric theory \citep{bickel1993efficient,van2003unified,tsiatis2006semiparametric, kennedy2016semiparametric}. This set of tools is often referred to as the doubly robust machine learning (DRML) framework \citep{chernozhukov2018double,hernan2020causal, kennedy2022semiparametric}. 

Following the framework in \citet{chernozhukov2018double}, we outline DRML methods for estimation of the \textit{local} average treatment effect from IV designs in which the so-called monotonicity assumption --- a commonly invoked assumption for the identification of IV effects (see Section~\ref{section:causal_estimand} for details) --- holds. We also summarize the use of sample splitting to ensure parametric convergence rates and valid inference for the target parameter. This provides investigators with an analytic pipeline for using ML methods for IV estimation to alleviate bias from model misspecification. In addition, we develop new, flexible methods for the estimation of related conditional effects within IV designs. More specifically, we extend a meta-learner framework \citep{kennedy2020optimal} to conditional local average treatment effect estimation.  

Additionally, we introduce new DRML methods related to the monotonicity assumption. First, we develop a method for profiling the principal strata that are implied by the monotonicity assumption. Next, we outline a new method of sensitivity analysis to probe the plausibility of the monotonicity assumption. Importantly, our sensitivity analysis framework allows investigators to judge whether study conclusions are robust to possible violations of the monotonicity assumption.

\subsection{Related Work}
There has been a great deal of interest in applying doubly robust or machine learning-based methods in IV designs, but little that has incorporated both sets of ideas. Here, we briefly review related work to better place our contributions in the context of the literature. First, some existing work has focused on applying semiparametric theory to develop robust estimators of causal effects in IV settings. \citet{ogburn2015doubly} proposed doubly robust estimators of marginal and conditional local average treatment effects under an IV design with monotonicity. Their approach presupposes a parametric model for the conditional effects, and in addition, they focus on parametric models for instrument probabilities and other nuisance functions. By contrast, we pursue a nonparametric approach where all necessary components of the observed data distribution can be flexibly estimated, e.g., with machine learning algorithms. \citet{abadie2003semiparametric} proposed semiparametric estimators of parameters corresponding to least squares projections of local conditional effects. \citet{tan2010marginal} proposed doubly robust estimators of parameters of marginal structural models and structural nested mean models, employing parametric models for the component nuisance functions. Finally, \citet{liu2020identification} derived necessary and sufficient conditions for identification of the average treatment effect on the treated in the presence of an IV, and proposed doubly robust estimators of this target estimand, again under parametric models for the nuisance functions involved. 

Second, there have been many applications of flexible models or machine learning algorithms for causal effect estimation in IV designs, without formal nonparametric efficiency guarantees for estimation of target parameters. In particular, investigators have employed regression trees and bayesian causal forests \citep{bargagli2019heterogeneous, bargagli2020causal, spanbauer2022flexible}, kernel-based approaches \citep{singh2019kernel, muandet2020dual}, and deep neural networks \citep{hartford2017deep, lewis2018adversarial, bennett2019deep} to estimate causal quantities in IV settings. See \citet{wu2022instrumental} for a review of this area of research.

Lastly, there has been some work to combine ideas from semiparametric theory and machine learning to develop efficient and robust estimators of conditional and marginal causal effects on the basis of instrumental variables. This work closely aligns with \citet{chernozhukov2018double}, particularly concerning the estimation and inference of marginal effects in IV settings. Their work outlines the DRML procedure, which directly informs our methodology. \citet{mauro2020instrumental} employed DRML methods for estimation of certain dynamic treatment effect estimands in an IV design. \citet{kennedy2019robust} considered a generalization of monotonicity in continuous IV settings, and proposed DRML-based estimates of the local instrumental variable curve. \citet{diazordaz2020data} developed DRML-based estimators of projection-based conditional local effects. That is, asserting a ``working'' parametric model for these conditional effects (e.g., thinking of these as closest parametric approximations to true conditional effects), they developed a targeted maximum likelihood estimator (TMLE) of the finite set of parameters in the working model. \citet{lee2020doubly} proposed a DRML-based estimator of the local average treatment effect for a time-to-event outcome under right censoring. Finally, similar to one of our proposals, \citet{frauen2022estimating} developed a doubly robust meta-learner of local effects conditional on all covariates. 

As such, the existing research on DRML methods for the IV framework is generally focused on more specific applications, and hasn't yet outlined a complete set of DRML methods for more standard IV applications. In addition, we propose a doubly robust meta-learner that differs in that it involves two pseudo-outcomes and second-stage regressions, and we also allow for more coarsened analysis by developing estimators of effects conditional on subsets of covariates. The rest of the article is organized as follows. In Section~\ref{sec:iv}, we review notation, assumptions, and the target marginal and conditional causal estimands in our study using IV methods. In Section~\ref{sec:drml}, we develop DRML estimation methods for both marginal and conditional treatment effects. In Section~\ref{sec:mono}, we review our proposed methods for profiling and sensitivity analysis. Next, we perform a simulation study to demonstrate how our proposed methods improve over extant methods in Section~\ref{sec:sim}. Finally, we present empirical results in Section~\ref{sec:results} and concluding remarks in Section~\ref{sec:conc}.

\section{The IV Design}
\label{sec:iv}

\subsection{Notation and Setup}
In this section, we introduce the notation to describe our statistical setting. We define each observational data unit as $O:=(Y, A, Z, X)$ where $Y$ represents the outcome of interest, $A$ is a binary treatment, $Z$ is a binary instrument and $X = (X_1,\ldots,X_k)$ is a set of baseline covariates. Specifically, in our study, $Y$ is a binary indicator for an adverse event. Treatment $A=1$ indicates that a patient received operative care after their emergency admission, while $A = 0$ indicates that a patient received alternative non-operative care. The instrument $Z$ is an indicator that takes the value of $1$ when the continuous measure of a surgeon's tendency to operate (TTO) is above the sample median, and 0 otherwise. That is, an IV value of 1 indicates that a patient has been assigned to a high TTO surgeon and thus are more likely to receive operative care. We observe an independent and identically distributed sample $(O_1, ..., O_n)$ from an unknown data-generating distribution $P_0$, and we denote by $\mathbb{P}_n$ its corresponding empirical distribution. We index objects by $P$ when they depend on a generic distribution, and we use subscript $0$ as short-hand for $P_0$. For instance, we denote the expectation under $P_0$ by $\E_0$. For a distribution $P$ and any $P$-integrable function $\eta$, we define $P\eta = \int \eta \, dP$. For instance, $\mathbb{P}_n \eta$ represents $n^{-1}\sum_{i=1}^n \eta(O_i)$. Occasionally, we slightly abuse the notation and treat the random variable as an identity mapping to itself, such that, $\mathbb{P}_n Y = n^{-1}\sum_{i=1}^n Y_i$.

We use the Neyman-Rubin potential outcome framework to formalize assumptions for the identification of the target causal parameters \citep{neyman1923application, rubin1974estimating}. We use $Y(a)$ to denote the potential outcome that \emph{would have been observed} had treatment been set to value $A=a$. When the treatment-outcome relationship is confounded (e.g., when $A$ is not randomized), estimating the effect of $A$ on $Y$ is challenging, relying on measurement of a sufficient set of confounders. However, if a valid instrumemt $Z$ is available, a causal effect of $A$ on $Y$ can be estimated consistently even in the presence of unobserved confounders. With an instrument, we require two additional potential outcomes: $Y(z,a)$ the potential outcome that would have been observed had the IV been set to $Z=z$ and treatment been set to $A=a$, and $A(z)$ the potential treatment status if the IV had been set to $Z=z$.

\subsection{Core IV Assumptions}
Next, we formalize the three core IV assumptions using the potential outcome framework, and discuss their plausibility in the context of our application. First, our notation implicitly assumes the stable unit treatment value assumption (SUTVA) \citep{Rubin:1986}: both the potential treatments received $A(z)$ and potential outcomes $Y(z,a)$, for $z, a=0,1$, depend solely on the value $z$ of the instrument, and $a$ of the treatment for $Y(z,a)$, for each individual. That is, there is only one version of the instrument and treatment, and $A_i(z)$ and $Y_i(z,a)$ are not affected by the value of $(Z_{i'}, A_{i'})$ for $i'\neq i$. These two components of SUTVA are often referred to as the consistency and no-interference assumptions, respectively. 

As discussed earlier, a valid instrument $Z$ must satisfy the following three conditions: (1) $Z$ affects or is associated with $A$, (2) is as good as randomized, possibly after conditioning on measured covariates $X$, and (3) affects the outcome $Y$ only indirectly through $A$ \citep{Angrist:1996}. The three core IV assumptions can be stated as follows:

\begin{enumerate}[label=\textbf{(A\arabic*)},leftmargin=2cm]
\item \textbf{Relevance:} \label{as:relevance}$Z$ is associated with exposure $A$, or equivalently, $P_0\big(A(1) = A(0)\big) \neq 1$.
\item \label{as:unconfounded}\textbf{Effective random assignment:} For all $z,a \in \{0,1\}$, $Z \indep \big(A(z), Y(z)\big) \mid X$. This is sometimes called the ``unconfoundedness'' assumption. 
\item \label{as:excl_restriction}\textbf{Exclusion restriction:} $Y(z,a) = Y(z',a)$ for all $z,z',a \in \{0,1\}$. This implies that $Z$ only affects $Y$ through its influence on $A$ i.e., $Y(z, a) \equiv Y(a)$, for all $z,a \in \{0,1\}$. 
\end{enumerate}

In our study, assumption \ref{as:relevance} implies that TTO affects the likelihood of a patient receiving surgery, which has been empirically validated in previous work \citep{keeleegsiv2018}. Although neither \ref{as:unconfounded} nor \ref{as:excl_restriction} can be tested, assumption \ref{as:unconfounded} is plausible as our data is based on emergency admissions, and patients may have little choice in selecting surgeons with a high or low TTO. We also adjust for baseline covariates in our analysis to further bolster this assumption. Statistical adjustment for baseline covariates is key motivation for developing more flexible methods of estimation. Assumption \ref{as:excl_restriction} implies that any effect of a surgeon's TTO on the outcome must only be a consequence of the medical effects of surgery. This can be violated in our study if treatment from a surgeon with a high TTO includes other aspects of medical care that can affect the outcome. This type of violation is relatively unlikely since nursing care in the emergency department is the same regardless of surgeon. To bolster the plausibility of this assumption, we compare patients within the same hospital that were assigned to surgeons with different levels of TTO, which holds other systemic factors of care constant \citep{keeleegsiv2018}.  \citet{keeleexrest2018} conduct a falsification test for the exclusion restriction based on TTO, and find no evidence against this assumption.
 
\subsection{Monotonicity and Local Average Treatment Effect}\label{section:causal_estimand}

One common statistical estimand of interest in an IV analysis is the following quantity:
\begin{equation}
\label{eq:LATE_identification}
  \chi_0 := \frac{\mathbb{E}_0\big[\mathbb{E}_0\big(Y \mid Z=1, X\big)-\mathbb{E}_0\big(Y \mid Z=0, X\big)\big]}{\mathbb{E}_0\big[\mathbb{E}_0\big(A \mid Z=1, X\big)-\mathbb{E}_0\big(A \mid Z=0, X\big)\big]}.
\end{equation}
The seminal work by \citet{Angrist:1994} shows that this estimand identifies a specific \textit{subgroup} causal effect. For this identification result to hold, we must introduce an additional assumption called \textit{monotonicity}. Typically, the monotonicity assumption is interpreted by stratifying the study observations into four so-called principal strata~\citep{Angrist:1996,Frangakis:2002}. That is, when both the instrument and treatment are binary, we can stratify observations into four groups or principal strata: \textit{compliers} (i.e., $A(1) > A(0)$), \textit{always-takers} (i.e., $A(1) = A(0) = 1$), \textit{never-takers} (i.e., $A(1) = A(0)=0$) and \textit{defiers} (i.e., $A(1) < A(0)$) \citep{angrist1996identification,Frangakis:2002}. In our application, always-takers will receive surgery regardless of the surgeons' TTO, while never-takers never receive surgery regardless of a surgeons' TTO. Compliers, however, receive surgery because they are assigned to a high TTO surgeon. Defiers always act contrary to the assignment of the IV. Monotonicity is then defined as follows:
\begin{enumerate}[label=\textbf{(A\arabic*)},leftmargin=2cm]
\setcounter{enumi}{3}
\item \label{as:monotonicity}\textbf{Monotonicity:} $P_0\big(A(1) < A(0)\big) = 0$.
\end{enumerate}
In other words, the monotonicity rules out the existence of defiers under the data-generating distribution $P_0$. The fundamental result from \citet{Angrist:1994} states that, under suitable assumptions, the IV estimand given by \eqref{eq:LATE_identification} equals a causal effect among compliers, or \textit{local average treatment effect} (LATE).
\begin{lemma}[Identification of LATE]\label{lm:id_late}
Assuming SUTVA, \ref{as:relevance}--\ref{as:monotonicity}, and $0 < P_0(Z = 1 \mid X) < 1$ almost surely, then 
\[\chi_0 = \E_0\big[Y(1) - Y(0) \mid A(1) > A(0) \big].\]
\end{lemma}

In our application, the LATE represents the effect of surgery among those patients who receive surgery because they receive care from a high TTO surgeon. It is critical to note that the LATE may not coincide with a more common target causal estimand, such as the average treatment effect (ATE). The ATE, represented by $\mathbb{E}_0\big[Y(1) - Y(0)\big]$, measures the average difference in outcomes when all individuals in the study population are assigned to the surgery group versus when all individuals are assigned to the non-operative management group. While the ATE provides a summary of treatment effects over the entire population, the LATE only represents the effect among compliers.

The monotonicity assumption may be controversial in the context of preference-based instruments like TTO \citep{swanson2014think,swanson2017challenging}. Monotonicity violations may arise when the instrument is not delivered uniformly to all subjects. In our context, the monotonicity assumption would not hold if there patients that are treated contrary to a physician's preferences for surgery. Indeed, a lack of such contrary cases may seem unlikely since a surgeon's preferences weigh a variety of risks and benefits. As an example, assume we have two surgeons who work in the same hospital. Surgeon 1 generally prefers to operate but makes exceptions for cholecystitis patients (because of some known contraindications). Surgeon 2 generally tries to avoid surgical treatments but makes exceptions for patients who are very healthy and can withstand the invasive nature of many procedures. Thus any patient that has cholecystitis but is also very healthy would be treated in a way contrary to each surgeons' general preferences for operative care and is a defier. Under this scenario, monotonicity does not hold. Given this concern, we propose and develop a sensitivity analysis to probe the plausibility of the monotonicity assumption.

\subsection{The Conditional Local Average Treatment Effect}

In the literature using IV methods for EGS conditions, the primary focus has been on the LATE. While the LATE is a useful summary of the overall effect among compliers, it tells us little about possible variation of the treatment effect across patients. To explore the possibility of heterogeneous treatment effects, we can estimate covariate-specific LATEs or conditional LATEs (CLATEs).  In our application, the CLATE corresponds to the causal effect of surgery on patient outcomes for subgroups of complying patients defined by their baseline covariates, such as age and comorbidities. By estimating CLATEs, we can investigate how the causal effect of surgery varies for different subgroups of patients, which can inform clinical decision-making and help identify which patients are either most likely to benefit from the surgical intervention or be at risk of an adverse event. This information is especially important in cases where the overall average effect of the treatment may be misleading due to heterogeneity in treatment effects across different patient subgroups.

To identify the CLATE, we invoke the identification conditions outlined in \citet{abadie2003semiparametric} and \citet{ogburn2015doubly}. First, we define the covariate set $V$, which is a subset of $X$. \citet{abadie2003semiparametric} first studied the special case $V = X$. The covariates in $V$ are the set of covariates identified as effect modifiers. We then define the conditional relevance assumption, which is a conditional analog to \ref{as:relevance}:
\begin{enumerate}[label=\textbf{(A'\arabic*)},leftmargin=2cm]
\item \textbf{Conditional relevance:} \label{as:cond_relevance}$Z$ is associated with exposure $A$ conditioning on $V=v_0$, or equivalently, $P_0\big(A(1) = A(0) \mid V=v_0\big)\neq 1$.
\end{enumerate}
Replacing \ref{as:relevance} with \ref{as:cond_relevance} leads to the 
following identification result, as shown in \citet{ogburn2015doubly}:
\begin{lemma}[Identification of Conditional LATE]
Assuming SUTVA, \ref{as:cond_relevance}, \ref{as:unconfounded}, \ref{as:excl_restriction}, \ref{as:monotonicity}, and $0 < P_0(Z = 1 \mid X) < 1$ almost surely, then 
\begin{align}
    \chi_0(v_0) &:= \frac{\mathbb{E}_0\big[\mathbb{E}_0\big(Y \mid Z=1, X\big)-\mathbb{E}_0\big(Y \mid Z=0, X\big)\mid V=v_0\big]}{\mathbb{E}_0\big[\mathbb{E}_0\big(A \mid Z=1, X\big)-\mathbb{E}_0\big(A \mid Z=0, X\big)\mid V=v_0\big]}\nonumber \\
    &= \E_0\big[Y(1) - Y(0) \mid A(1) > A(0), V=v_0\big].\nonumber
\end{align}

\end{lemma}

The statistical functional of interest, $\chi_0(v_0)$, is now the ratio of two regression functions. While the standard LATE is a scalar quantity, the CLATE is a function of $V$, making estimation and inference significantly more complex. Below, we outline how to use a meta-learner framework for flexible estimation of the CLATE. Recent research has developed meta-learner methods for conditional treatment effects that incorporate semiparametric theory \citep{luedtke2016super,wager2018estimation,foster2019orthogonal,kennedy2020optimal,chernozhukov2018double,nie2021quasi}. In Section~\ref{sec:DR_learner}, we adapt the ``DR-Learner'' from \citet{kennedy2020optimal} to the context of CLATE estimation.

\section{DRML Estimation and Inference for IV Designs}
\label{sec:drml}

The DRML framework combines semiparametric theory, machine learning methods, and sample splitting to derive flexible yet efficient nonparametric estimators of causal parameters \citep{bickel1993efficient,van2003unified,tsiatis2006semiparametric, kennedy2016semiparametric}. A crucial feature of the DRML framework is the use of influence functions --- a central object of semiparametric theory --- to construct root-$n$ consistent estimators under relatively weak conditions. In particular, the framework provides general heuristics to derive estimators that converge in distribution to a mean-zero normal distribution with the smallest asymptotic variance among a large class of estimators. 

First, we briefly review how to estimate the average treatment effect (ATE) using the DRML framework. See \citet{kennedy2022semiparametric} for a more detailed treatment. Under causal positivity, consistency, and no unmeasured confounding, the ATE can be written as the contrast of two averaged regression functions: $\Gamma_0 := \E_0[\E_0(Y \mid A = 1, X)]-\E_0[\E_0(Y \mid A = 0, X)]$, where $Y$ is the outcome, $A$ is binary treatment, and $X$ are measured confounders. Estimating $\E_0(Y \mid A = a, X)$ for $a \in \{0,1\}$ is a standard regression problem and flexible ML methods may be preferred to avoid model misspecification. However, if $\Gamma_0$ is estimated as the average of predictions from an ML model, the estimator typically inherits first order bias and converges slower than the so-called parametric rate of $n^{-1/2}$. Moreover, depending on the ML method, statistical inferences may not be valid. 

Using semiparametric theory, it is possible to construct an estimator $\widehat\Gamma$ of $\Gamma_0$ on the basis of some mean-zero function $\dot{\Gamma}_0^*$, such that 
\[
n^{1/2}\left(\widehat\Gamma - \Gamma_0\right) \overset{d}{\longrightarrow} N\left(0, P_0\dot{\Gamma}_0^{*2}\right)
\]
\noindent
where $P_0\dot{\Gamma}_0^{*2}$ is the asymptotic variance of $\widehat{\Gamma}$, scaled by sample size $n$. The \textit{efficient influence function}, which is the optimal $\dot{\Gamma}_0^{*}$ in terms of asymptotic efficiency in a nonparametric model, is based on two regression models: the outcome model $\E_0(Y \mid A, X)$ and the propensity model $P_0(A = 1 \mid X)$. The outcome model can be estimated by regressing the outcome against the treatment and all confounders; while the propensity model can be estimated by regressing the treatment against all confounders. The final estimate of the ATE is obtained by combining estimates from these two models. This approach is often called ``doubly robust'' as it is consistent if either the outcome model or the propensity model is correctly specified \citep{Scharfstein:1999a, robins2001inference}, and it typically exhibits a faster rate of convergence due to second-order product bias compared to using a single model. 

To avoid complicated statistical dependencies from using the same data for estimation and selection of tuning parameters, sample-splitting or cross-fitting is then applied in the estimation process \citep{robins2008higher, zheng2010asymptotic, chernozhukov2018double}. Overall, the DRML framework enables the estimation of treatment effects using ML methods while potentially retaining fast parametric-type $n^{-1/2}$ inferential rates.

\subsection{DRML for LATE}\label{sec:nonparametric-late}

In this section, we outline the DRML framework for the estimation of LATEs under weak distributional assumptions. To begin, we define two functions:
\begin{equation}\label{eq:gamma_def}
    \gamma_P(x) := \mathbb{E}_P(Y \mid X=x, Z = 1) - \mathbb{E}_P(Y \mid X=x, Z = 0)
\end{equation}
\noindent and 
\begin{equation}\label{eq:delta_def}
    \delta_P(x) := \mathbb{E}_P(A \mid X=x, Z = 1) - \mathbb{E}_P(A \mid X=x, Z = 0).
\end{equation}
The first term is the effect of the IV on $Y$ conditional on covariates, and the second is the effect of the IV on $A$ conditional on the covariates. Lemma~\ref{lm:id_late} established that the LATE is identified via the parameter $\chi_0 := \Gamma_0/\Delta_0$, where $\Gamma_0 := P_0\gamma_0$ and $\Delta_0 :=P_0\delta_0$. Note that the relevance, effective random assignment, and monotonicity assumptions ensure that $\Delta_0$ is non-zero. Based on this identification result, we motivate a plug-in estimator $\widehat{\chi} := \widehat{\Gamma}/\widehat{\Delta}$ that is the ratio of two estimators. 

The efficient influence functions of $\Gamma_P$ and $\Delta_P$ at $P$ in a nonparametric model are equivalent to those for the ATE
and are \citep{hahn1998role}:
  \[\dot{\Gamma}^*_P := (y,z,x) \mapsto \frac{2z - 1}
    {\pi_{P}(x,z)}\{y - \mu_{P}(x,z)\} +
    \gamma_P(x) - \Gamma_P\]
    \noindent and 
    \[ \dot{\Delta}^*_P :=(a,z,x) \mapsto
    \frac{2z - 1}{\pi_{P}(x,z)}\{a -
    \lambda_{P}(x,z)\} + \delta_P(x) - \Delta_P\]
These two influence functions involve three nuisance functions: $\mu_{P}(x,z) := \mathbb{E}_P(Y \mid X=x, Z = z)$, $\lambda_{P}(x,z) := \mathbb{E}_P(A \mid X=x, Z = z)$, and $\pi_{P}(x,z) := P(Z = z \mid X=x)$. The first is an outcome model, the second is a model for treatment, and the third is a model for the IV. We discuss estimation of these nuisance functions below. Since influence functions are mean-zero by definition, we often consider ``uncentered" influence functions whose mean corresponds to the target estimand. To distinguish uncentered influence functions from the original influence functions, we use the notation $\dot{\Gamma}_P$ without the asterisk. The uncentered influence functions of $\Gamma_P$ and $\Delta_P$ are given by 
\begin{equation}\label{eq:gamma_uncentered}
    \dot{\Gamma}_P := (y,z,x) \mapsto \frac{2z - 1}
    {\pi_{P}(x,z)}\{y - \mu_{P}(x,z)\} +
    \gamma_P(x)
\end{equation}
\noindent
and
\begin{equation}\label{eq:delta_uncentered}
    \dot{\Delta}_P := (a,z,x) \mapsto \frac{2z - 1}
    {\pi_{P}(x,z)}\{a - \lambda_{P}(x,z)\} +
    \delta_P(x).
\end{equation}
The function $\dot{\Gamma}_0$ has the property that $P_0\dot{\Gamma}_0 = \Gamma_0$ and similarly, $P_0\dot{\Delta}_0 = \Delta_0$ for $\dot{\Delta}_0$.

The estimators $\widehat\Gamma$ and $\widehat\Delta$ we propose are ``one-step'' estimators, corresponding to the sample mean of the estimated uncentered influence functions. To construct the estimator, we first need to estimate the unknown nuisance functions: $\mu_{P}$, $\lambda_{P}$ and $\pi_P$. In the DRML framework, the analyst selects the estimation method for these nuisance functions. To avoid model misspecification, one could select flexible ML methods such as random forests or generalized boosting for this task. However, the estimation process is completely agnostic to the choice of nuisance estimation method. For example, one could estimate $\mu_P$, $\lambda_P$, or $\pi_P$ with least squares instead. In this case, the estimator would be doubly-robust but would also be fully parametric. As such, reducing bias in the DRML framework depends on selecting suitably flexible estimation methods for the nuisance functions.

 Next, we compute the ratio of the plug-in estimators of uncentered influence functions as follows:
\begin{equation}\label{eq:chi_hat_expression}
\widehat{\chi} := \widehat\Gamma/\widehat\Delta = \frac{\mathbb{P}_n\dot{\Gamma}_{\widehat{P}}}{\mathbb{P}_n\dot{\Delta}_{\widehat{P}}} = \frac{\mathbb{P}_n\left[\frac{2Z - 1}
      {\widehat{\pi}}\{Y -
      \widehat{\mu}\} +
      \widehat{\gamma}\right]}{\mathbb{P}_n\left[\frac{2Z - 1}{\widehat{\pi}}\{A -
      \widehat{\lambda}\} +
      \widehat{\delta}\right]},
\end{equation}
omitting inputs to nuisance estimates.
We introduced the subscript $\widehat{P}$ above to denote that all unknown quantities are replaced by their respective estimators. 

In the DRML framework, sample splitting and cross-fitting is typically part of the estimation process (e.g., \citet{chernozhukov2018double}). As such, we also pursue sample splitting, which will slightly modify the procedure described above. One straightforward sample-splitting scheme is to randomly split the study population into two samples $D_0$ and $D_1$. First, using $D_0$, we estimate the nuisance terms $\mu_0$, $\lambda_0$, and $\pi_0$ and any tuning parameters needed for the ML methods. Then, we compute an estimate of $\chi_0$ as in~\eqref{eq:chi_hat_expression} by plugging in data from $D_1$ into the previously fitted nuisance functions. Moreover, to increase efficiency, the role of $D_0$ and $D_1$ can be swapped and the final estimate can be taken as the average of the two estimates. In fact, more than two sample splits can be utilized for extra stability. This procedure is known as cross-fitting. In addition to using sample splitting for the estimation of the nuisance terms, we can use an ensemble of ML methods. That is, we can estimate each nuisance term with multiple or different ML methods and combine the results to further reduce the possibility of model misspecification. For instance, we can estimate $\mu_0$ using both random forests and boosting and then combine the two estimates by a meta-learner such as Super Learner \citep{van2007super}. 

One of the main benefits of semiparametric theory is that it facilitates the construction of $(1-\alpha)$-level confidence intervals for nonparametric estimators. The limiting variance of the resulting estimator coincides with the variance of the efficient influence function, divided by $n$. We first define the efficient influence function of $\chi_0$ in the following lemma:
\begin{lemma}\label{lemma:IF} 
Let $\chi_P = \Gamma_P/\Delta_P$ be the ratio of the two parameters $\Gamma_P$ and $\Delta_P$, whose influence functions are $\dot{\Gamma}^*_P$ and $\dot{\Delta}^*_P$, respectively. The influence function of $\chi_P$ at $P$ is given by
  \begin{align*}\label{eq:chi_influence_function}
      \dot{\chi}^*_P &:= (y,a,z,x) \mapsto\\
      &\quad \frac{1}{\Delta_P}\left(\frac{2z - 1}
      {\pi_{P}(x,z)} \left[y - \mu_{P}(x,z) -
        \chi_P\{a - \lambda_{P}(x,z)\}\right] +
      \gamma_P(x) - \chi_P\delta_P(x)\right).\nonumber
  \end{align*}
 In a nonparametric model, the influence function above is also the efficient influence function.
\end{lemma}
The proof of this Lemma is provided by \citet{kennedy2021semiparametric} as Example 6. The next result characterizes the asymptotic distribution of the estimator $\widehat{\chi}$ and provides a basis for statistical inference. We now introduce the following additional conditions: 
\begin{enumerate}[label=\textbf{(A\arabic*)},leftmargin=2cm]
\setcounter{enumi}{4}
\item $\|\widehat{\pi}-\pi_0\| \max(\|\widehat{\lambda}-\lambda_0\|, \|\widehat{\mu}-\mu_0\|) = o_P(n^{-1/2})$
\label{as:doubl-robust} 
\item $\|\widehat{\pi}-\pi_0\| =o_P(1)$, $\|\widehat{\mu}-\mu_0\| =o_P(1)$ and $\|\widehat{\lambda}-\lambda_0\| =o_P(1)$ 
\label{as:empirical-process} 
\item $|Y|$ is bounded almost surely
\label{as:bounded-y} 
\end{enumerate}
where $\|f\| := \left(P_0f^2\right)^{1/2}$ for any $P_0$-square-integrable function of $O$. Condition~\ref{as:doubl-robust} requires the nuisance estimators to converge at certain rate. and will be satisfied when all estimators converge faster than $n^{-1/4}$, but it also extends to more general settings. We are now ready to state the weak convergence of the estimator $\widehat\chi$.
\begin{lemma}
Assuming~\ref{as:doubl-robust}, \ref{as:empirical-process}, \ref{as:bounded-y}, $\Delta_0 > 0$, and there exist $\varepsilon_1, \varepsilon_2> 0$ such that $\varepsilon_1 < P_0(Z = 1 \mid X) < 1-\varepsilon_1$ and $\varepsilon_2 < \widehat\pi(X, 1), \pi(X, 1)< 1-\varepsilon_2$ almost surely, then 
    \begin{equation}
        n^{1/2}(\widehat\chi - \chi_0)\overset{d}{\longrightarrow} N\left(0, \,P_0\dot{\chi}^{*2}_0\right).\nonumber
    \end{equation}
\end{lemma}
We note that \ref{as:bounded-y} can be relaxed to a moment condition. The proof of this Lemma is provided in the Supplementary Material. With weak convergence of the estimator $\widehat\chi$ established, we can motivate the following $(1-\alpha)$-level Wald-style confidence interval based on its asymptotic normality:
\begin{equation}\label{eq:confidence_interval}
    \widehat\chi \pm n^{-1/2}q_{1-\alpha/2} \left(\mathbb{P}_n\dot{\chi}_{\widehat P}^{*2}\right)^{1/2}
\end{equation}
where $q_p$ is the $p$th quantile of the standard Normal distribution and $\dot{\chi}^*_{\widehat P}$ is an estimator of the influence function whose expression is given by Lemma~\ref{lemma:IF}. 

We summarize the DRML procedure outlined in this section in the following algorithm:
\begin{algorithm}[DRML Estimation Algorithm for LATE]
\label{algo:late}
Let $(D_0, D_1)$ denote two independent samples of $O := (Y, A, Z, X)$ with sample sizes $n_0$ and $n_1$. The estimator $\widehat{\chi}$ of $\chi_0$ and the confidence intervals based on the DRML framework are given as follows:
\begin{enumerate}
    \item Construct estimators $\widehat{\pi}$, $\widehat{\mu}$, and $\widehat{\lambda}$ using $n_0$ samples from $D_0$.
    \item For each $O_1, O_2, \dots O_{n_1} \in D_1$, construct the plug-in estimates of (uncentered) influence functions, corresponding to $\big\{\dot{\Gamma}_{\widehat{P}}(Y_i, Z_i, X_i)\big\}_{i=1}^{n_1}$ and $\big\{\dot{\Delta}_{\widehat{P}}(A_i, Z_i, X_i)\big\}_{i=1}^{n_1}$ given by \eqref{eq:gamma_uncentered} and \eqref{eq:delta_uncentered} using the estimators from step 1. 
    \item Compute ${n_1}^{-1}\sum_{i=1}^{n_1}\dot{\Gamma}_{\widehat{P}}(Y_i,Z_i,X_i)$ and ${n_1}^{-1}\sum_{i=1}^{n_1}\dot{\Delta}_{\widehat{P}}(A_i,Z_i,X_i)$ using $D_1$. Report the final estimator $\widehat{\chi}$ given by \eqref{eq:chi_hat_expression}.
    \item Construct the $(1-\alpha)$-level confidence intervals using \eqref{eq:confidence_interval}.
\end{enumerate}
\end{algorithm}

\subsection{DR-Learner for Conditional LATE}
\label{sec:DR_learner}

Next, we develop a DRML estimator for the conditional LATE. Unlike the estimator we developed for the LATE, it is not straightforward to extend the DRML framework to estimation of the CLATE. This is primarily because the CLATE is an infinite-dimensional object and does not possess the necessary pathwise differentiability property required for the existence of an influence function \citep{kennedy2022semiparametric}. To address this challenge, we adapt the DR-Learner from \citet{kennedy2020optimal} to our IV setting.  We now outline a meta-learner algorithm for the estimation of CLATEs:

\begin{algorithm}[DR-Learner for Conditional LATE]
\label{algo:conditional_l}
Let $(D_0, D_1)$ denote two independent samples of observations $O := (Y, Z, A, X)$ with sample sizes $n_0$ and $n_1$. Let $V \subseteq X$ be a covariate set. The estimator $\widehat\chi(v_0)$ of $\chi_0(v_0)$ for any $V=v_0$ is given as follows:
\begin{enumerate}
    \item Construct estimators $\widehat{\pi}$, $\widehat{\mu}$, and $\widehat{\lambda}$ using $n_0$ samples from $D_0$.
    \item For $O_1, O_2, \dots O_{n_1} \in D_1$, construct a plug-in estimate of (uncentered) influence functions $\big\{\dot{\Gamma}_{\widehat{P}}(Y_i, Z_i, X_i)\big\}_{i=1}^{n_1}$ and $\big\{\dot{\Delta}_{\widehat{P}}(A_i, Z_i, X_i)\big\}_{i=1}^{n_1}$ using the nuisance estimators from step 1. 
    \item Regress $\big\{\dot{\Gamma}_{\widehat{P}}(Y_i, Z_i, X_i)\big\}_{i=1}^{n_1}$ and $\big\{\dot{\Delta}_{\widehat{P}}(A_i, Z_i, X_i)\big\}_{i=1}^{n_1}$ on $V$ using $D_1$, yielding two estimators $\widehat{\Gamma}(v)$ and $\widehat{\Delta}(v)$ of the following regression functions:
    \begin{align*}
        \Gamma_0(v) &:= \E_0\big[\dot{\Gamma}_0(Y,Z,X) \mid V = v\big]\, \text{ and }\\
        \Delta_0(v) &:= \E_0\big[\dot{\Delta}_0(A,Z,X) \mid V = v\big].
    \end{align*}
    \item Report $\widehat{\Gamma}(v_0)/\widehat{\Delta}(v_0)$ as an estimate of $\chi_0(v_0)$.
    \item Repeat steps 2 to 4 using bootstrap resampled data $\widetilde{D}_1$ from $D_1$ and report the $\alpha/2$th and $(1-\alpha/2)$th quantiles of the estimator.
\end{enumerate}
\end{algorithm}
The procedure for estimating CLATEs in Algorithm~\ref{algo:conditional_l} differs from the one for estimating LATEs in Algorithm~\ref{algo:late} in a crucial way. Step 3 involves an additional regression step, where a set of so-called psuedo-outcomes (i.e., the estimated uncentered influence functions) are regressed on the effect modifiers $V$, and the resulting estimator is the ratio of the estimated conditional expectation evaluated at $v_0$. Note that this additional regression step can be done with any nonparametric method. As before, so long as flexible estimation methods are used for the nuisance functions, this method is fully nonparametric.

Nonparametric inference for irregular functionals with no influence function is an active area of research \citep{bibaut2017data, takatsu2022debiased}. Although the literature typically studies the limiting distribution for specific estimators, the behavior of our proposed estimator heavily depends on the particular regression function estimator used in step 3 of Algorithm~\ref{algo:conditional_l}. Here, we use the bootstrap to estimate the variance of the resulting estimator. Under some conditions, this approach works with many choices of the regression procedure used in step 3. Specifically, we repeat steps 2--4 in Algorithm~\ref{algo:conditional_l} over the bootstrap samples of $D_1$. The $(1 - \alpha)$-level bootstrap confidence interval is given by the $\alpha/2$th and $(1-\alpha/2)$th quantiles of the estimator based on bootstrap samples. 

We note that traditional bootstrap theory requires repeating the entire process from steps 1--4. However, it is practically infeasible to repeat step 1 as it typically involves fitting complex ML methods with large data, which in our case is time consuming. Establishing formal theoretical requirements for the bootstrap procedure when step 1 is computationally expensive is an interesting avenue for future research. In this work, we proceed by assuming that the nuisance estimation errors from step 1 are negligible compared to the bootstrap errors from steps 2--4, which is plausible for larger sample sizes. 

\section{Nonparametric Profiling and Sensitivity Analysis}
\label{sec:mono}

 In this section, we develop DRML-based methods for two key analytic tasks related to the monotonicity assumption: profiling and sensitivity analysis. 

\subsection{Profiling}

Profiling, which refers to calculating descriptive statistics for the complier, always-taker, and never-taker populations, is often used to infer whether the LATE might differ in magnitude from the ATE. This has been recommended as a necessary part of any LATE-focused IV analysis \citep{swanson2013commentary,baiocchi2014instrumental}. For instance, we can calculate whether the healthiest patients are likely to be always-takers, and the sickest patients are never-takers. In such case, the LATE most likely differs from the ATE, as operative care would likely have different effects among these sub-populations.


The concept of profiling was first outlined for compliers in \citet{angrist2008mostly}. Specifically, \citet{angrist2008mostly} use Abadie's $\kappa$-weights from semiparametric instrumental variable methods to estimate complier profiles. \citet{baiocchi2014instrumental} extended complier profiling methods and introduced the technique to health service research and epidemiology. Later research extended profiling to the always-taker and never-taker subpopulations \citep{marbach2020profiling}. The analytic goal of profiling is to compare the covariate distributions for the complier, always-taker and never-taker populations to the overall patient population profile. If profiling reveals that the complier subpopulation is different with respect to observable covariates that are likely to be predictive of treatment effect magnitude, this suggests that the LATE may not be close to the ATE. 

We propose DRML-based methods for nonparametric density estimation of the covariate distributions for compliers, always-takers, and never-takers. We first present the method for a discrete random variable $V$. Assuming SUTVA, \ref{as:relevance}, \ref{as:unconfounded}, \ref{as:monotonicity} and $0 < P(Z=1\mid X) < 1$ almost surely, we can identify the density for baseline covariate at $V=v_0$ among the compliers as follows:
\begin{align}
    P(V=v_0 \mid A(1) > A(0)) =\frac{\E_P\big[\delta_P(X)\mathds{1}(V=v_0)\big]}{\E_P\big[\delta_P(X)\big]}\label{eq:functional_complier}
\end{align}
We provide a derivation of the above identification result in the Supplementary Material. We note that \ref{as:excl_restriction} is not required for the identification above. The denominator of the estimand~\eqref{eq:functional_complier} is identical to that of the standard LATE, which can be estimated as $\mathbb{P}_n\dot{\Delta}_{\widehat{P}}$. The uncentered influence function for the numerator of \eqref{eq:functional_complier} is simply given by $(y, a, z, x') \mapsto \dot{\Delta}_0(a,z,x')I(v'=v_0)$. Here, $v'$ is the component of $x'$ corresponding to $V \subseteq X$. This follows from so-called the product rule of influence functions; Example 2 of \cite{kennedy2022semiparametric} provides an analogous derivation for the case with $V=X$. The DRML-based estimator of \eqref{eq:functional_complier} is then as follows:
\begin{equation}
\frac{\mathbb{P}_n\dot{\Delta}_{\widehat{P}}\mathds{1}\big(V=v_0\big)}{\mathbb{P}_n\dot{\Delta}_{\widehat{P}}} = \frac{\mathbb{P}_n\left[\left(\frac{2Z - 1}
      {\widehat{\pi}}\left\{A -
      \widehat{\lambda}\right\} +
      \widehat{\delta} \right)\mathds{1}\big(V=v_0\big)\right]}{\mathbb{P}_n\left[\frac{2Z - 1}{\widehat{\pi}}\left\{A -
      \widehat{\lambda}\right\} +
      \widehat{\delta}\right]}.
\end{equation}
This estimator is also consistent when $\mathbb{P}_n\dot{\Delta}_{\widehat{P}}$ is consistent in view of the continuous mapping theorem, and thus inherits the double-robust property of the ATE estimator $\widehat{\Delta}$. 

To construct $(1-\alpha)$-level confidence intervals, we can use Lemma~\ref{lemma:IF} in the context of our estimator except $\dot{\Gamma}^*_P$ and ${\Gamma}_P$ are now replaced with $\dot{\Delta}_P(A, Z, X)\mathds{1}(V=v_0)-\E_P\big[\delta_P(X)\mathds{1}(V=v_0)\big]$ and $E_P[\delta_P(X)\mathds{1}(V=v_0)]$, respectively.


Similarly for always-takers and never-takers, we have the following identification results:
\begin{align}
    P(V=v_0 \mid A(1) = A(0)=1) 
    &= \frac{\E_P\big[\lambda_P(X, 0)\mathds{1}(V=v_0)\big]}{\E_P\big[\lambda_P(X, 0)\big]}\nonumber
\end{align}
\noindent and 
\begin{align}
    P(V=v_0 \mid A(1) = A(0)=0) = \frac{\E_P\big[\left\{1-\lambda_P(X, 1)\right\}\mathds{1}(V=v_0)\big]}{\E_P\big[1-\lambda_P(X, 1)\big]}\nonumber.
\end{align}

For a continuous covariate $V$, we propose a kernel density estimator of $p(v_0 \mid A(1) > A(0))$:
\begin{equation}
\frac{\mathbb{P}_n\dot{\Delta}_{\widehat{P}}K_{v_0,h}}{\mathbb{P}_n\dot{\Delta}_{\widehat{P}}}\, \quad\text{where}\quad  \, K_{v_0,h}:= v\mapsto\frac{1}{h}K\left(\frac{v-v_0}{h}\right).
\end{equation}
Here, $K$ is a non-negative kernel function and $h>0$ is a so-called smoothing bandwidth parameter. This estimator is consistent assuming that $h \longrightarrow 0$, $nh \longrightarrow \infty$ and the second derivative of the true density function is bounded. Despite the consistency of the estimator, valid statistical inference for this parameter is challenging to obtain. When the bandwidth parameter is selected to achieve the optimal rate of convergence of the estimator, the smoothing-bias becomes asymptotically non-negligible, so that the confidence intervals are no longer centered at the target parameter \citep{wasserman2006all}. Alternatively, \citet{kennedy2021semiparametric} studies the estimation of a covariate-adjusted density function under a semiparametric framework. They consider the projection of the target density function onto a working model with respect to a chosen metric. This methodology can be explored in the context of profiling as a future work.


\subsection{Sensitivity Analysis}

In Section~\ref{section:causal_estimand}, we raised concerns about the validity of the monotonicity assumption in our study. We outlined one plausible scenario where defiers may exist: when a surgeon's preferences represent a consideration of various risks and benefits, cases may arise where a patient is treated differently from a physician's overall preference, resulting in defiers.

To address the potential violation of monotonicity, we conduct a novel sensitivity analysis to examine the robustness of our estimates to such violations. \citet{Angrist:1994} demonstrated that when monotonicity is violated, the IV estimand can be written as the sum of the standard IV estimand and a term that depends on two parameters, $\delta_1$ and $\delta_2$:
\begin{equation}
    \mathbb{E}_0[Y(1)-Y(0) \mid A(1) >  A(0)]=\chi_0 +\frac{\delta_1\delta_2}{\Delta_0}.
\end{equation}
where $\delta_1$ and $\delta_2$ are defined as:
\begin{align}
    \delta_1 & := P_0(A(1) < A(0))\\
    \delta_2 & :=\mathbb{E}_0[Y(1)-Y(0) \mid A(1) < A(0)]-\E_0[Y(1)-Y(0) \mid A(1) > A(0)].
\end{align}
In words, $\delta_1$ represents the proportion of defiers, a value between $0$ and $1$, and $\delta_2$, on the other hand, represents the difference in average treatment effects between the defiers and compliers. When the outcome is binary, $\delta_2$ must take a value between $-2$ and $2$. The above expression recovers the standard identification result of the LATE when either $\delta_1$ or $\delta_2$ is zero. On the other hand, values of  $\delta_1$ or $\delta_2$ that are large in magnitude may imply a discrepancy between the standard IV estimand and the LATE. 

We can estimate the standard IV estimand and $\Delta_0$ using the DRML-based methods and treat $\delta_1$ and $\delta_2$ as free parameters in our sensitivity analysis. We now propose a simple analysis based on the mapping:
\begin{equation}
    \widehat\xi := (\delta_1, \delta_2) \mapsto \widehat{\chi} +\frac{\delta_1\delta_2}{\mathbb{P}_n\dot{\Delta}_{\widehat{P}}}
\end{equation}
which provides the estimated difference between the LATE and the IV estimand. This mapping is consistent when $\widehat{\chi}$ is consistent, implying the double robustness of the sensitivity mapping.

In our application, we plot the estimated upper or lower bounds on the LATE as a joint function of both $\delta_1$ and $\delta_2$, and find a set of $\delta_1$ and $\delta_2$ such that $\left\{\delta_1, \delta_2: \widehat\xi(\delta_1, \delta_2) = 0\right\}$. These values correspond to the minimal violation of monotonicity resulting in a sign change in the estimated LATE. Based on subject matter knowledge, we then evaluate whether the values of $\delta_1$ and $\delta_2$ at this frontier appear larger or smaller relative to plausible values for the study. Moreover, given a range of plausible values for $\delta_2$, we can estimate and interpret the magnitude of $\delta_1$ necessary to render $\widehat\xi$ zero.

\section{Simulation Study}
\label{sec:sim}

Next, we investigate the finite-sample properties of the DRML IV estimator in a simulation study. Our aim is to compare the DRML IV estimator to the standard TSLS estimator across a range of scenarios to understand the comparative performance of each method. We analyze three different scenarios that capture different levels of nonlinearity in the data generating process (DGP). We also vary the sample size to the probe the large-sample performance of the DRML IV estimator compared to TSLS. The results from the first scenario are reported in this section, while the results from the other two scenarios are provided in Supplementary Material. First, we describe the common elements of the DGP across the scenarios we consider.  

Across all the simulation scenarios, we generate a set of baseline covariates: $X = (X_1, X_2)$, where $X_1 \sim \mathrm{Unif}(-1,1)$, $X_2 \sim \mathrm{Bernoulli}(0.3)$, such that $X_1 \independent X_2$, and an unobserved covariate $U \sim \mathrm{Unif}(-1.5, 1.5)$ such that $U \independent X$. Generation of $Z$, $A$, and $Y$ then proceeds on the basis of various functional forms for $\pi_0(x, z) = P_0(Z = z \mid X = x)$, $\lambda_0^{\dagger}(x, z, u) = \mathbb{E}_0(A \mid X = x, Z = z, U = u)$, and $\mu_0^{\dagger}(x, a, u) = \mathbb{E}_0(Y \mid X = x, A = a, U = u)$. Specifically, we generate the instrument as $Z \sim \mathrm{Bernoulli}(\pi_0(X, 1))$, the treatment as $A \sim \mathrm{Bernoulli}(\lambda_0^{\dagger}(X, Z, U))$, and the outcome as $Y \sim N(\mu_0^{\dagger}(X, A, U), 0.2^2)$. For the specification of $\mu_0^{\dagger}$, we adopt an additive model $\mu_0^{\dagger}(X, A, U) = r(X)A + s(X) + 1.5 U$, where $r(X)$ represents the conditional average treatment effect, and $s(X)$ the baseline mean of the outcome.
For each simulation scenario, described in detail below, we vary the sample size $n$ and the functional forms of $\pi_0$, $\lambda_0^{\dagger}$, $r$ and $s$. We note that, with respect to the observed data distribution, any model $\lambda_0^{\dagger}$ is consistent with the monotoncity assumption. That is, we could strictly enforce monotonicity by generating
$A(0) \sim \mathrm{Bernoulli}(\lambda_0^{\dagger}(X, 0, U))$, then 
\[A(1) = A(0) + (1 - A(0))\mathrm{Bernoulli}\left(\frac{\lambda_0^{\dagger}(X, 1, U) - \lambda_0^{\dagger}(X, 0, U)}{1 - \lambda_0^{\dagger}(X, 0, U)}\right),\]
and finally set $A = (1 - Z)A(0) + ZA(1)$, but this also results in $\mathbb{E}_0(A \mid X, Z, U) = \lambda_0^{\dagger}(X, Z, U)$. However, explicit generation of $(A(0), A(1))$ as above is useful for computing the true underlying LATE in the scenarios below: the LATE is the mean of $r(X)$ given $A(1) > A(0)$.
 
For each scenario, we compute and evaluate three estimators of the LATE. First, we use the TSLS estimator using the \texttt{AER} package in \texttt{R}, which assumes the main effects of $X_1, X_2$ are linear. Next, we use the DRML IV estimator, which is based on estimates of the nuisance functions $\pi_0$, $\lambda_0$, $\mu_0$ --- the latter two are induced by $\lambda_0^{\dagger}$, $\mu_0^{\dagger}$, and the distribution of $U$. We employ one DRML IV estimator that uses parametric models for the nuisance functions, including only linear main effects for $X$ in each. We then employ a second DRML IV estimator where we estimate the nuisance functions using an ensemble of random forest fits. Specifically, the ensemble is based on fits from \texttt{rpart}, \texttt{rpartPrune}, and \texttt{glm} using the \texttt{SuperLearner} package in \texttt{R}. This set of estimators allows us to isolate instances in which flexible estimation of the nuisance functions is critical to reduce model misspecification. We refer to these as the parametric and nonparametric DRML estimators respectively. For all three methods, we record bias and root-mean squared error (RMSE). 


In the first scenario, the DGP departs from linear models with nonlinearity across a number of the models in the DGP. Specifically, we use the following functional forms for the models in the DGP: 
\begin{itemize}
\item
  $\pi_0(X, 1) = \mathrm{expit}(0.4 X_1 - 0.8X_2 + 0.4
  \mathds{1}(X_1 > 0))$
\item
  $\lambda_0^{\dagger}(X, Z, U) = \mathrm{expit}(-0.3 - 0.4X_1 - 0.14 X_2
  + 1.1 Z - 0.55 X_1 Z - 0.7 \mathds{1}(X_1 > 0) + 0.7 U)$
\item
  $r(\boldsymbol{X}) = - 4 X_1 + 6 \{X_2 - 0.3\} - 4
  \{\mathds{1}(X_1 > 0) - 0.5\}$
\item $s(\boldsymbol{X}) = 40 - 7X_1 - 8X_2 + 10\mathds{1}(X_1 > 0)$
\end{itemize}
We compute the true LATE using the quantities in the DGP as described earlier. In this DGP, there is a threshold effect for $X_1$ in each of the models. The nonparametric DRML IV estimator uses flexible estimation methods for the nuisance functions, and this should reduce bias from model misspecfication. However, we expect the parametric DRML IV estimator and TSLS to be biased due to model misspecification. Additionally, we are also interested in whether the nonparametric DRML IV estimator remains efficient as the sample size grows. 

Figure~\ref{fig:scenario1} contains the results from the simulation under scenario 1. The nonparametric DRML IV estimator does display some minor bias, however, the bias decreases with larger sample sizes. In theory, we should be able to reduce this bias further by including additional learners into the ensemble. Notably, the bias for the nonparametric DRML IV estimator is less than that from either of the parametric methods. In fact, the bias for TSLS is almost three times larger than for the nonparametric DRML IV estimator. Critically, these results highlight the need to use flexible methods of estimation for the nuisance functions. The parametric DRML IV estimator does reduce bias compared to TSLS, but not nearly to the extent of the nonparametric DRML IV estimator.

Next, we review the results for RMSE. When the sample size is smaller, here 500, the difference in RMSE between TSLS and DRML IV is smaller. However, once the sample size is 1000, the nonparametric DRML IV estimator clearly dominates the other two methods in terms of RMSE, with additional gains being possible with larger sample sizes. This result is not surprising given the level of bias present for the parametric methods of estimation.

\begin{figure}[htbp]
  \centering
  \includegraphics[width=0.9\linewidth]{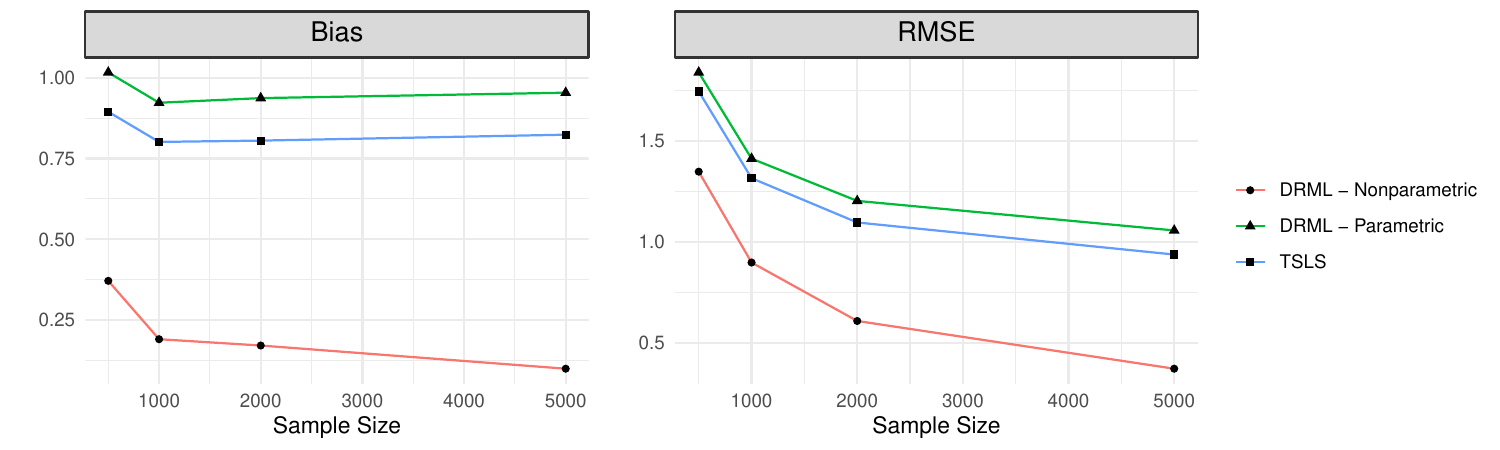}
    \caption{Comparing bias and RMSE for TSLS and DRML IV under simulation scenario 1.}
  \label{fig:scenario1}
\end{figure}

\section{The Assessment of the Efficacy of Surgery for Cholecystitis}
\label{sec:results}
\subsection{DRML IV Estimates}
     
We now apply our methods to the analysis of whether operative management is effective for patients with cholecystitis. First, we provide a more in-depth discussion of the data and specifically the instrument.  As we noted above, the data include a small set of demographic measures including indicators for race, sex, age, and insurance type. Next, the data include measures of baseline patient frailty. The measures of frailty are indicators for severe sepsis or septic shock and pre-existing disability. Next, there are indicators for 31 comorbidities based on Elixhauser indices \citep{elixhauser1998comorbidity}. Comorbidities are prior medical conditions that may complicate treatment for cholecystitis. These comorbidities include conditions such as anemia, hypertension, and obesity. Finally, we adjust for hospital level effects using fixed effects for each hospital, which is equivalent to conducting a within-hospital analysis. We consider such a within-hospital approach an essential part of the analysis, since it implies that we are comparing patients who receive care in the same hospital, which should hold fixed other systemic factors of care that might affect outcomes. In our judgement, this increases the plausibility of unconfoundedness and the exclusion restriction.

We measure the IV in our study based on the percentage of times a surgeon operates when presented with an EGS condition (TTO). We plot the TTO distribution in Figure~\ref{fig:tto}. While many surgeons operate more than 50\% of the time, there is wide variation in the TTO distribution. For a surgeon at the average of the TTO distribution, he or she operates 67\% of the time. 
Next, we split TTO at the median and calculated the percentage of times a patient received surgery for cholecystitis for the two categories. We find that patients that received care from a surgeon with a TTO below the median had an operation 45\% of the time. For patients that received care from a surgeon with a TTO above the median, he or she had an operation 95\% of the time.

\begin{figure}[htbp]
  \centering
    \includegraphics[scale=0.5]{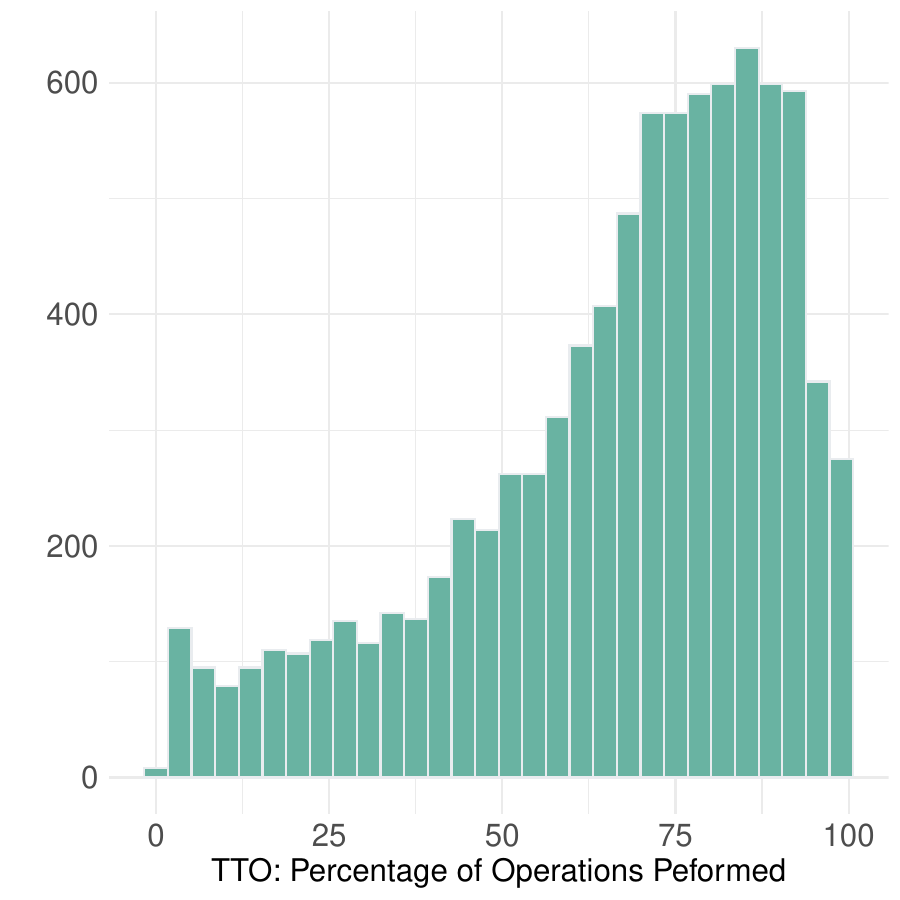}
    \caption{The distribution of tendency to operate measured as the percentage of times a surgeon operates.}
  \label{fig:tto}
\end{figure}

Our primary outcome is an adverse event following treatment. The measure of adverse events is an indicator that is 1 if a patient either died or had a prolonged length of stay. Prolonged length of stay (PLOS) is defined as hospital and operation-specific length of stay being greater than the 75th percentile. 

The first phase in the analysis is estimating the effect of surgery on adverse outcomes using IV methods, i.e., targeting the LATE. In this step of the analysis, we contrast DRML methods with more standard parametric methods. Next, we profile the key principal strata and review the likelihood of treatment by complier-status interactions. We then explore whether the treatment effects are heterogenous and seek to understand whether any variation is associated with key baseline covariates.

\begin{figure}[htbp]
  \centering
    \includegraphics[scale=0.40]{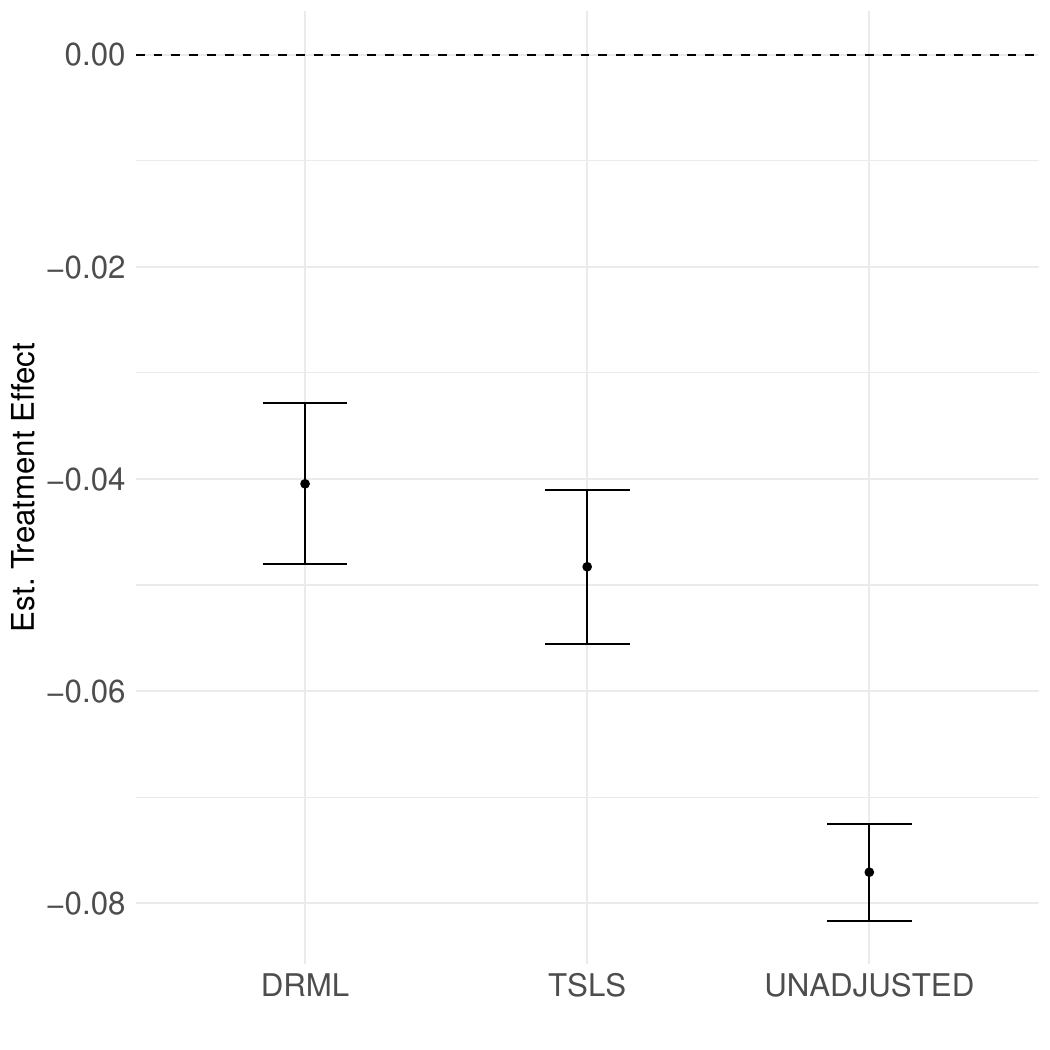}
    \caption{Estimates for the effect of surgery on adverse outcomes for patients with cholecystitis.}
  \label{fig:outs}
\end{figure}

For our DRML estimator of the LATE, one key practical decision necessary for the applied analysis is the choice of ML methods for estimation of the nuisance functions. We can select either a single learner such as a random forest or instead we can use an ensemble of learners. In our analysis, we use an ensemble of three different learners: a generalized linear model, a generalized additive model, and a random forest.  Figure~\ref{fig:outs} displays the results for three different estimation strategies: unadjusted, adjustment via a parametric model (TSLS), and adjustment via our DRML estimator. For the unadjusted analysis, we estimate the risk of an adverse event as function of being exposed to surgery without controlling for any baseline covariates. In the unadjusted analysis, we find that the risk of an adverse outcome is nearly 8\% lower for patients that had surgery compared to patients that underwent medical management. Given the invasive nature of surgery, it may be the case that healthier patients are more likely to be selected for surgery. This bias may well be in operation, given that the two adjusted IV estimates are nearly 50\% smaller than the unadjusted estimate. That is, once we adjust for observed covariates and account for unobserved confounders with a LATE-focused IV approach, we find that the effect of surgery is still beneficial --- the risk of an adverse outcome is approximately 4\% lower for the compliers. That is, the LATE in our context is the effect of surgery for those patients that had surgery they had a surgeon with a higher preference for operative care. We explore how the compliers differ from always-takers and never-takers next. In addition, the confidence intervals do not include zero. The DRML IV estimates do not differ significantly from estimates based on the fully parametric TSLS. In brief, for any given analysis, the additional flexibility of ML methods may not change the results compared to parametric methods. However, the additional flexibility protects against possible bias from model misspecification.

\subsection{Profiling}

Next, we conduct a profiling analysis to describe the characteristics of patients in each of the key principal strata. Figure~\ref{fig:profile} contains covariate profiles for compliers, always-takers, and never-takers for a subset of key covariates. Across many of the covariates the estimates for the never-taker subpopulation tend to be imprecisely estimated due to limited sample size. However, it is often the case that the always-taker and complier populations differ in important respects. For example, always takers are much less likely to be septic compared to compliers. In addition, they are less likely to have some form of disability or be Medicare patients. Overall, this pattern is consistent with other work that has found that for condition like these always-takers tend to be healthier than compliers, while never-takers tend to display greater frailty \citep{kaufman2022operative}.

\begin{figure}[htbp]
  \centering
    \includegraphics[scale=0.5]{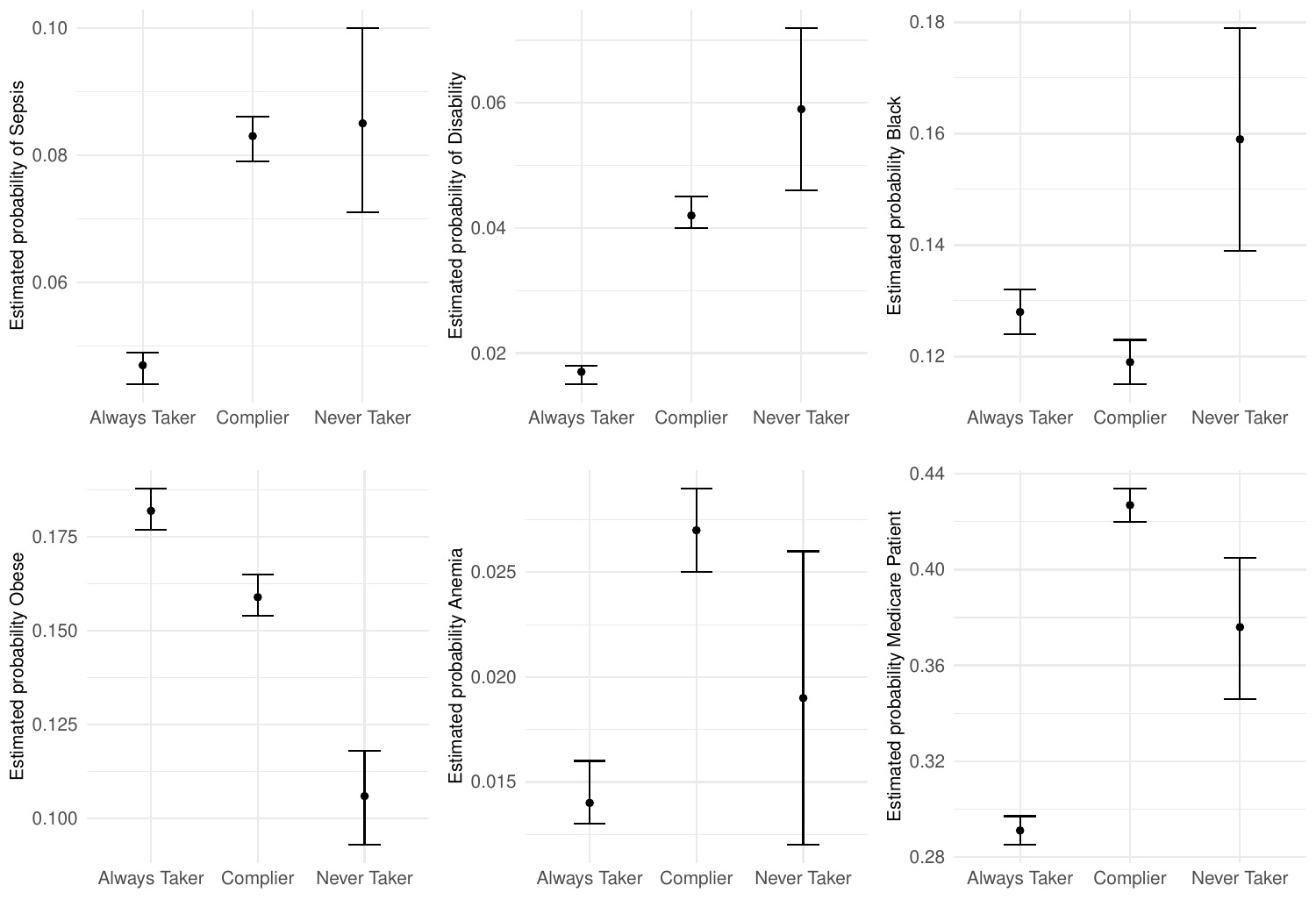}
    \caption{Covariate profiles for compliers, always-takers, and never-takers for selected baseline covariates: indicators for sepsis, disability, Black, obese, anemia, and medicare insurance}
  \label{fig:profile}
\end{figure}

Next, we conduct a more refined profiling analysis for age. Extant profiling analysis approaches based on parametric methods would estimate the average age among each of the principal strata. Our DRML IV methods allow us to estimate a smoothed density for continuous covariates by principal strata. Figure~\ref{fig:age-profile} contains the smoothed age density for compliers, always-takers, and never-takers. In the first panel, we observe that always-takers are mostly likely to be patients between 30 and 60. That is, younger patients are more likely to be surgical patients---most likely due to the fact that it is assumed they are more likely to tolerate the invasive nature of surgery. The smoothed density for compliers rises with age until it levels off around age 50. As such, we observe that a surgeons TTO tends to be decisive for older patients.

\begin{figure}[htbp]
  \centering
    \includegraphics[scale=0.5]{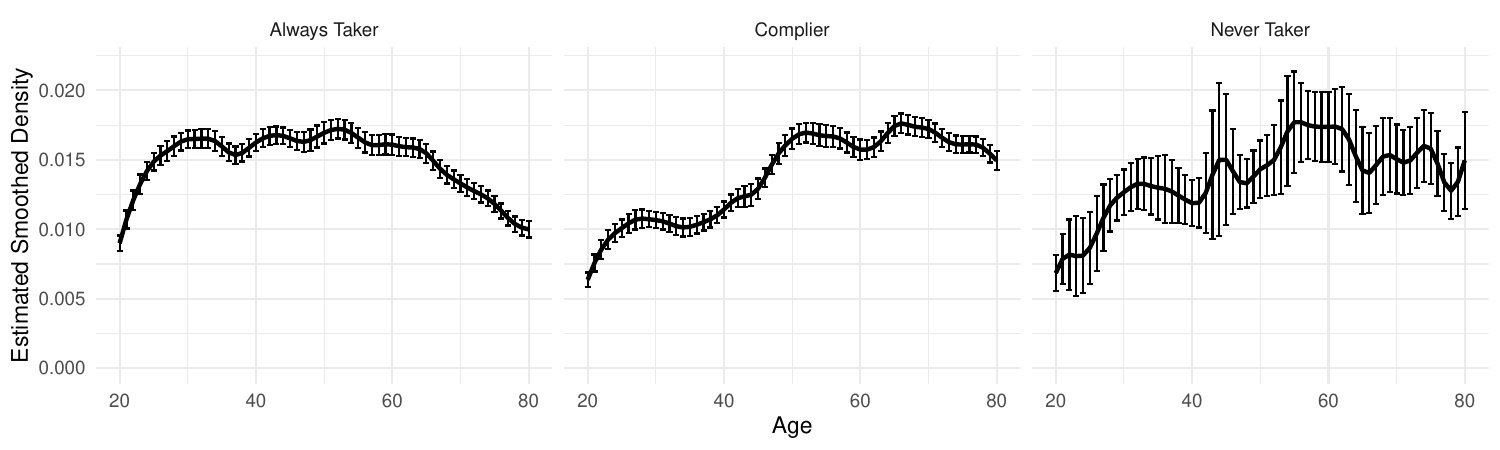}
    \caption{Estimated smoothed density for age by principal strata.}
  \label{fig:age-profile}
\end{figure}

\subsection{Heterogenous Treatment Effects}

Next, we consider the possibility that the effect of surgery varies by key patient subgroups. In our analysis, we focus on three key baseline covariates: age, sepsis, and the number of comorbidities. We identified these three covariates as important effect modifiers, since the effectiveness of surgery is often a function of baseline health. Patients that are generally healthy before they undergo surgery tend to respond better to surgical treatment. As such, older patients, patients with multiple comorbidities, and those with sepsis may be at higher risk for an adverse event even if surgery is generally effective on average.

As we outlined above, the DR-Learner approach to heterogenous treatment effect estimation provides estimates of individual level treatment effects (ITEs) that vary with effect modifiers.  The ITEs are then aggregated by the key effect modifiers to allow us to observe whether the treatment effect of surgery varies by that covariate. We begin with Figure~\ref{fig:ite} which displays the distribution of ITEs for the patient population in the study. The distribution of the ITEs provides some insight into the overall level of effect heterogeneity in our analysis. In Figure~\ref{fig:ite}, the dashed line marks the sample average, which corresponds to the DRML estimate reported in Figure~\ref{fig:outs}. We observe that a large part of the distribution is shifted below zero which represents a beneficial effect of surgery. However, the average effect also hides considerable variation. That is, we also observe that a substantial mass of the distribution of individual treatment effects is concentrated above zero, which implies that for these patients surgery caused an adverse outcome. 

\begin{figure}[htbp]
  \centering
    \includegraphics[scale=0.45]{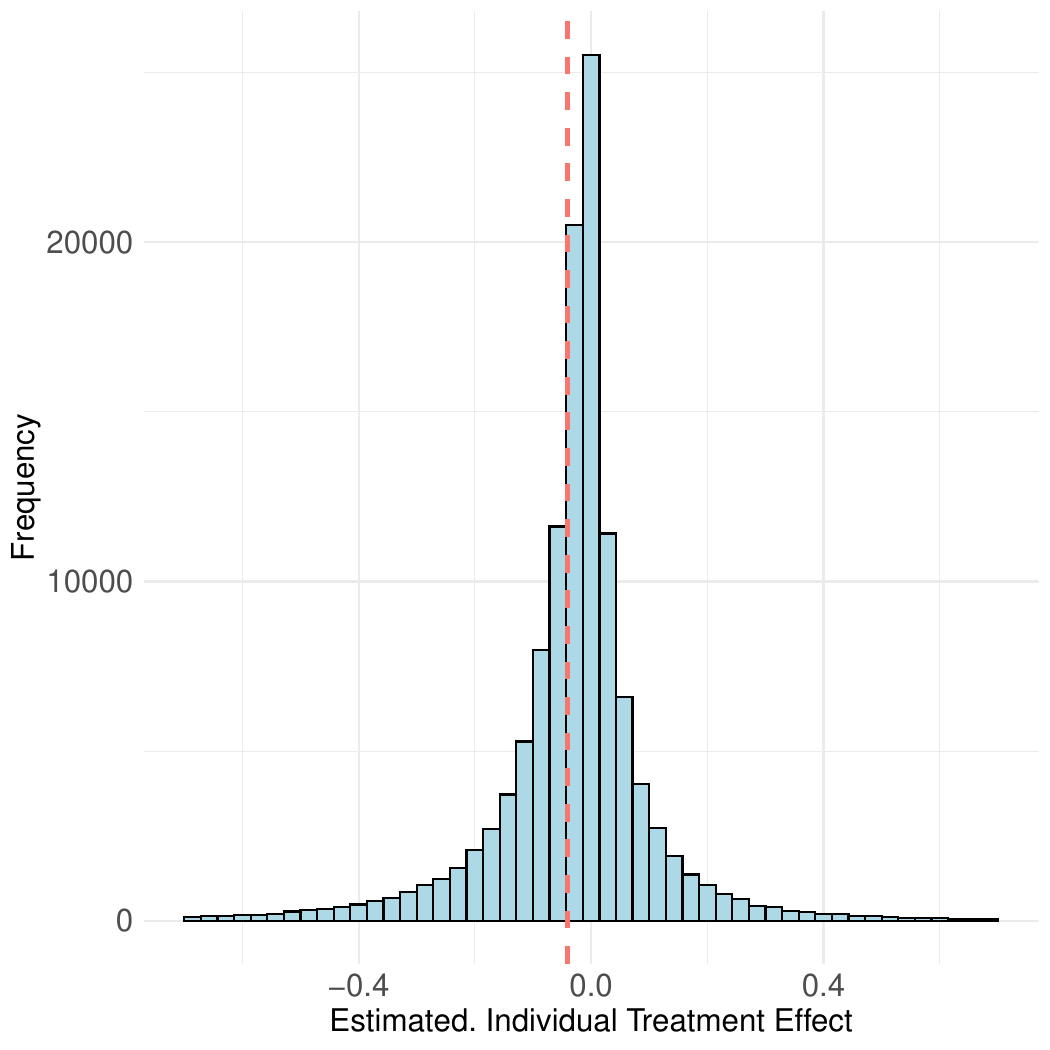}
    \caption{Distribution of individuals treatment effects for the risk of adverse event after surgery.}
  \label{fig:ite}
\end{figure}

Next, we seek to investigate whether specific baseline covariates explain variation in the individual level treatment effects. To that end, we aggregate the individual level treatment effects to estimate conditional average treatment effects for specific patient subgroups. Critically, our method does not impose any linearity assumptions on how the estimated effects vary as a function of a multi-valued effect modifier such as age or the number of comorbidities. First, we use the number of comorbidities to estimate CLATEs. We use the number of comorbidities, since it is often a strong predictor of adverse outcomes. Figure~\ref{fig:cate1} contains the average treatment effect conditional on the number of comorbidities. Based on the results in this figure, it does not appear that the number of comorbidities explain the variation in treatment effects. That is, the average effect within each subgroup is very close to if not identical to the estimated LATE.

\begin{figure}[htbp]
  \centering
    \includegraphics[scale=0.45]{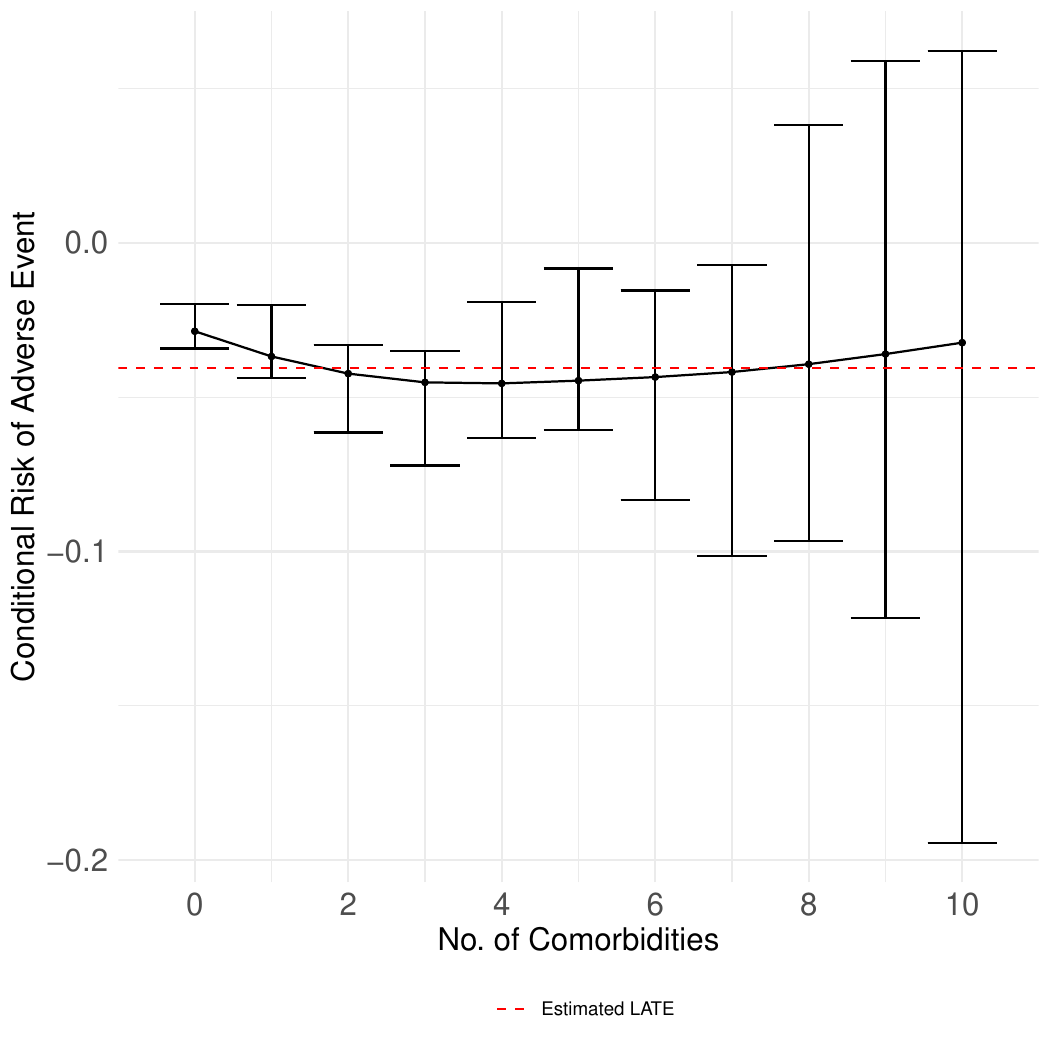}
    \caption{Conditional average treatment effects for surgery by the number of comorbidities.}
  \label{fig:cate1}
\end{figure}

Another key advantage of using DRML methods for treatment effect heterogeneity is that we can investigate whether effects vary by more complex patient subgroups. Next, we estimate how the treatment effects vary by sepsis, age, and the number of comorbidities. That, is we might expect older, septic patients with a higher number of comorbidities to be at the highest risk for an adverse effect. Figure~\ref{fig:cate2} contains two heatmaps that each display the CLATE by age, the number of comorbidities, and sepsis status. First, we find that septic patients are more likely to have an adverse outcome.  That is, for non-septic patients, we observe that the effect of surgery is minimal to protective, with the benefit increasing with age. For septic patients, the probability of an adverse event increases with the number of comorbidities.  Interestingly, older patients with a small number of comordibities also benefited from surgery. Thus, it appears that sepsis and comorbidities appear to be important effect modifiers. We stress that these results should be considered as exploratory. Validation of these results would require hypothesis testing with a second data source.

\begin{figure}[htbp]
  \centering
    \includegraphics[scale=0.6]{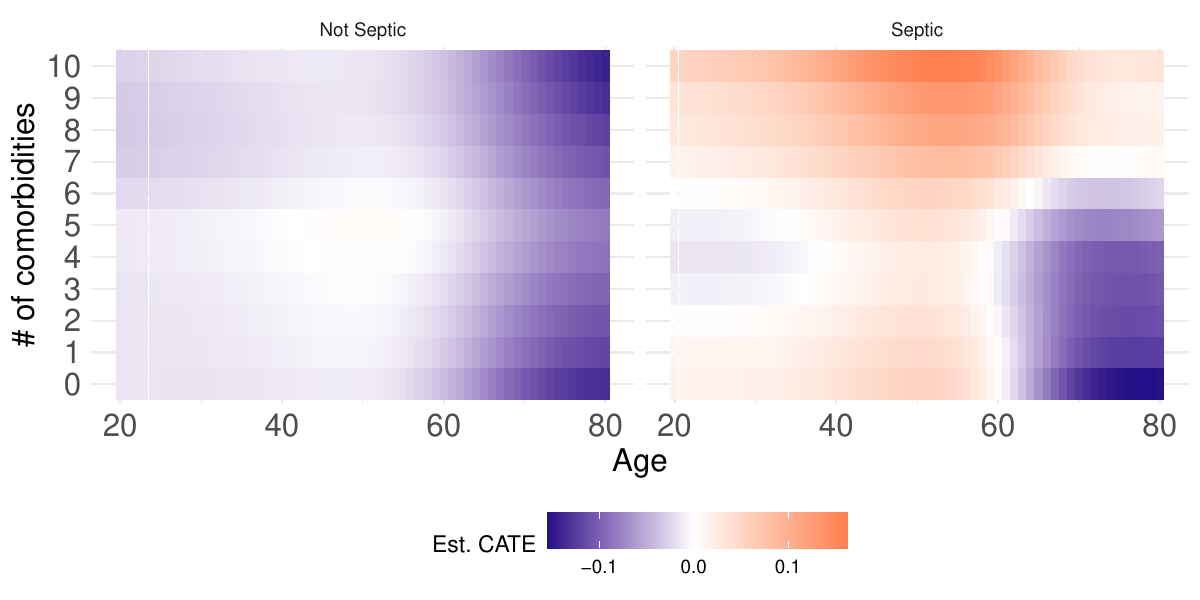}
    \caption{Conditional average treatment effects for surgery by sepsis, age, and the number of comorbidities.}
  \label{fig:cate2}
\end{figure}

\subsection{Sensitivity Analysis}

Lastly, we investigate whether our study results are robust to violations of the monotonicity assumption. As we outlined above, instruments like TTO may be prone to violations of the monotonicity assumption.  This can occur when a physician treats a patient contrary to his or her general preference for surgery. In our proposed sensitivity analysis, we plot how the upper bound for the estimated LATE varies as a function of the two sensitivity parameters: $\delta_1$, and $\delta_2$. We plot the results in Figure~\ref{fig:mono}. In Figure~\ref{fig:mono}, the x-axis corresponds to $\delta_1$ and y-axis corresponds to $\delta_2$. The dashed line represents the frontier where the upper bound on the LATE equals zero. 

In our context, this dashed line divides the plot into two regions. In the lower-left region, the upper bound on the LATE indicates that surgery is beneficial, since it reduces the risk of an adverse event. In the upper-right region, the upper bound indicates on the LATE indicates that surgery is potentially harmful, since a positive LATE corresponds to an increased risk of an adverse event with surgery. How do we use this plot to assess sensitivity?  We think it is helpful to focus in on a few values for $\delta_1$, and assess the corresponding ranges of $\delta_2$ values that are consistent with the sign of the original LATE estimate. For example, we note that if 25\% of the population were defiers, the sign on the LATE would be unchanged if the average difference in defier and complier risks is less than 10\%.  If the proportion of defiers is less than 5\%, the sign of the estimated LATE also remains unchanged if the average difference in defier and complier risks is 50\% or below.

\begin{figure}[htbp]
  \centering
    \includegraphics[scale=0.45]{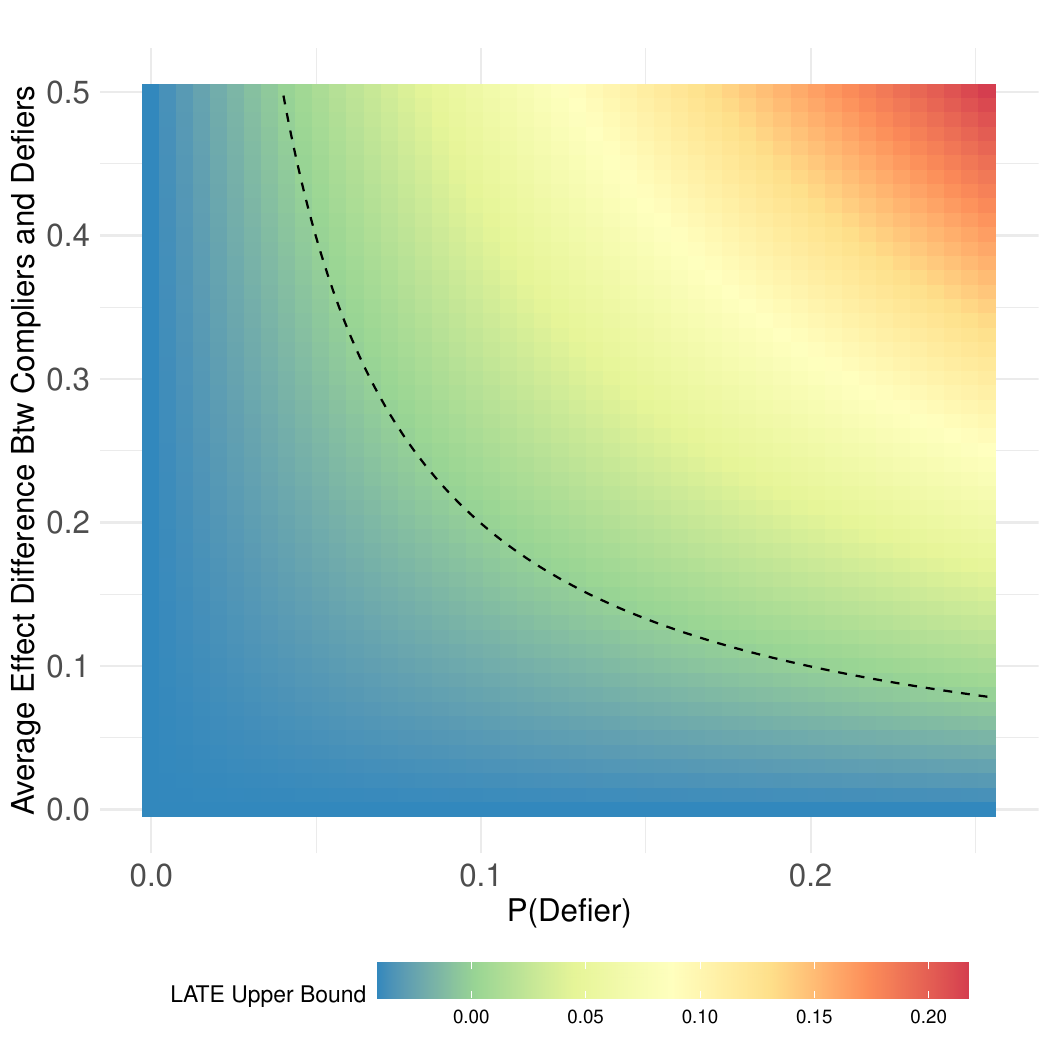}
    \caption{Upper bound on the estimate of LATE under the different levels of the monotonicity assumption violations.}
  \label{fig:mono}
\end{figure}

We observe that our primary estimate of the LATE (assuming monotonicity) is robust, with respect to its sign, over most reasonable values of the sensitivity parameters. For instance, we can allow up to 10\% of defiers if surgical treatment for them is up to 20\% more likely to induce adverse events compared to compliers. Based on these results, we may conclude that the sign of treatment effect is robust to reasonable violations of monotonicity. 

\section{Conclusion}
\label{sec:conc}

Methodologically, we offer a unified estimation framework for IV methods that combines ML methods to avoid bias from model misspecification, while maintaining established inferential properties. Our simulation study demonstrates he effectiveness of the DRML IV estimator in reducing bias and achieving parametric-like estimation error rates and inferential properties. Furthermore, we leveraged the DRML framework to develop novel methods for other aspects of IV designs. Specifically, we used DRML methods to profile compliance classes under principal stratification, and we extended the DR-Learner for heterogeneous effects to the IV framework. This enables us to examine the potential variation in the effect of surgery across different patients by assessing key effect modifiers. Finally, we developed a new method of sensitivity analysis targeted at the monotonicity assumption, which can be questionable in the case of IVs such as TTOs. Importantly, our findings display the robustness of our estimates when faced with fairly large deviations from this assumption.

Our study also contributes to the developing literature on optimal treatment strategies for EGS conditions. Given the ethical issues with randomization, the literature on this topic has been generally based on observational studies using IV methods. EGS conditions are heterogeneous, and, as a result, granular details are required about specific conditions to provide effective evidence-based guidance to clinicians. As such, various studies have focused on developing specific guidance for given EGS conditions and subgroups \citep{kaufman2022operative,moler2022local,hutchings2022effectiveness,keeleegsiv2018,
rosen2022analyzing}. Our study contributes to this body of evidence by focusing on cholecystitis --- inflammation of the gallbladder, which is one of the most common EGS conditions. Our finding indicates that, on average, operative management for cholecystitis led to lower rates of adverse events. Moreover, we found that the number of comorbidities was not predictive of variation in the effect of surgery.  However, age emerged as a strong effect modifier, as patients over the age of 70 seemed to be more susceptible to adverse events following surgery.

There are several avenues for future developments. Firstly, a broader set of simulations and applied analyses would be valuable in providing guidance on the most robust ensemble of learners across diverse settings. Secondly, more investigation is needed to elucidate the formal properties of the proposed DR-Learner of the CLATE under various methods for second-stage regression. Thirdly, it will also be of theoretical interest to determine minimax rates and corresponding optimal estimators of the CLATE, as was explored for the conditional average treatment effect in \citet{kennedy2022minimax}. Finally, it will be important to extend the ideas proposed in this paper to IV settings that do not assert monotonicity, moving beyond the CLATE and LATE, e.g., under certain structural models for the outcome, or alternatively constructing tight bounds on the (conditional) average treatment effect.

\section*{Acknowledgements}
The authors declare no conflicts. Research in this article was supported by the Patient-Centered Outcomes Research Institute (PCORI Awards ME-2021C1-22355) and the National Library of Medicine, \#1R01LM013361-01A1. All statements in this report, including its findings and conclusions, are solely those of the authors and do not necessarily represent the views of PCORI or its Methodology Committee. The dataset used for this study was purchased with a grant from the Society of American Gastrointestinal and Endoscopic Surgeons. Although the AMA Physician Masterfile data is the source of the raw physician data, the tables and tabulations were prepared by the authors and do not reflect the work of the AMA. The Pennsylvania Health Cost Containment Council (PHC4) is an independent state agency responsible for addressing the problems of escalating health costs, ensuring the quality of health care, and increasing access to health care for all citizens. While PHC4 has provided data for this study, PHC4 specifically disclaims responsibility for any analyses, interpretations or conclusions. Some of the data used to produce this publication was purchased from or provided by the New York State Department of Health (NYSDOH) Statewide Planning and Research Cooperative System (SPARCS). However, the conclusions derived, and views expressed herein are those of the author(s) and do not reflect the conclusions or views of NYSDOH. NYSDOH, its employees, officers, and agents make no representation, warranty or guarantee as to the accuracy, completeness, currency, or suitability of the information provided here. This publication was derived, in part, from a limited data set supplied by the Florida Agency for Health Care Administration (AHCA) which specifically disclaims responsibility for any analysis, interpretations, or conclusions that may be created as a result of the limited data set.

\begin{supplement}[id=suppA]
  \sname{Supplement A}
  \stitle{Complete Simulation Results}
  \slink[doi]{COMPLETED BY THE TYPESETTER}
  \sdatatype{.pdf}
  \sdescription{Additional Analytic Results}
\end{supplement}

\begin{supplement}[id=suppB]
  \sname{Supplement B}
  \stitle{Additional Empirical Results}
  \slink[doi]{COMPLETED BY THE TYPESETTER}
  \sdatatype{.pdf}
  \sdescription{Additional Simulation Results.}
\end{supplement}

\clearpage
\bibliographystyle{imsart-nameyear}
\bibliography{learner,IV-ML}

\end{document}


\noindent%
	{\bfseries\large
Supplement A: Additional Analytic Results}

\noindent%
	
	\noindent%

\pagenumbering{arabic}
\setcounter{equation}{0}
\setcounter{table}{0}
\setcounter{figure}{0}
\renewcommand{\theequation}{S\arabic{equation}}
\renewcommand{\thetable}{S\arabic{table}}
\renewcommand{\thefigure}{S\arabic{figure}}
\renewcommand\theassumption{S.\arabic{assumption}}
\renewcommand\thetheorem{S.\arabic{theorem}}

\section{Calculus of influence functions}
In this section, we use the notation $P$ to represent a generic distribution that belongs to a statistical model denoted as $\mathcal{P}$. Let $P_0 \in \mathcal{P}$ represent the unknown data-generating distribution of observations $O$, and let $\mathbb{P}_n$ denote its empirical distribution function associated with $\{O_1, \dots, O_n\}$. For any $P$-integrable function $f$, we define $Pf = \int f dP$. Specifically, we have $\mathbb{P}_n f=\frac{1}{n}\sum_{i=1}^n f(O_i)$. We also adopt the notation $\|f\|= (P_0 f^2)^{1/2}$. For $a, b \in \mathbb{R}$, we define $a \wedge b = \min(a, b)$.

We consider $\psi_1, \psi_2$ as real-valued functionals $\mathcal{P} \mapsto \mathbb{R}$ where $\mathcal{P}$ is a collection of distributions or a nonparametric model. We assume that these functionals are differentiable in a suitable sense ~\citep{van2002semiparametric}, allowing for the existence of corresponding influence functions. We denote the influence functions as $\dot \varphi_1^*$ and $\dot \varphi_2^*$, and the uncentered influence functions as $\dot \varphi_1$ and $\dot \varphi_2$, meaning that $P_0\dot \varphi_j=\psi_j(P_0)$ for $j = 1, 2$. Define an estimator of $\psi_j(P_0) = P_0 \dot{\varphi}_j$ given by $\widehat{\psi}_j = \mathbb{P}_n \widehat{\varphi}_j$, the sample mean of estimated uncentered influence functions. Then $\widehat{\psi}_j$ satisfies the following error decomposition:
    \begin{align*}
       \widehat{\psi}_j - \psi_j(P_0) =  \mathbb{P}_n \widehat{\varphi}_j +\mathbb{P}_n \dot{\varphi}_j-\mathbb{P}_n \dot{\varphi}_j- P_0 \dot{\varphi}_j=\mathbb{P}_n \dot{\varphi}^*_j + R_1+ R_2 
   \end{align*}
where 
\begin{align*}
    R_1 := (\mathbb{P}_n-P_0)(\widehat{\varphi}_j-\dot{\varphi}_j)\, \quad \text{and} \quad R_2:=P_0(\widehat{\varphi}_j-\dot{\varphi}_j).
\end{align*}
Suppose for now that both remainder terms $R_1$ and $R_2$ are $o_P(n^{-1/2})$, then the estimator is asymptotically linear:
\begin{align*}
    \mathbb{P}_n \widehat{\varphi}_j - P_0 \dot{\varphi}_j =\mathbb{P}_n \dot{\varphi}^*_j + o_P(n^{-1/2}).
\end{align*}
This representation immediately establishes the consistency of the estimator by the weak law of large numbers. Moreover, if we assume that the variance of $\dot{\varphi}^*_j$ is finite, we can establish the asymptotic normality of the estimator in view of the central limit theorem. Consequently, this informs the construction of confidence intervals that are asymptotically exact.

The first remainder term $R_1$ is commonly referred to as an empirical process term. When the estimator $\widehat{\varphi}_j$ is constructed using independent data from $\mathbb{P}_n$, for instance by sample-splitting or cross-fitting, Lemma 2 of \citet{kennedy2020sharp} states that: 
\begin{align*}
    R_1 = O_P(n^{-1/2}\|\widehat{\varphi}_j-\dot{\varphi}_j\|).
\end{align*}
This term becomes $o_P(n^{-1/2})$ if $\|\widehat{\varphi}_j-\dot{\varphi}_j\|=o_P(1)$, or in other words, $\widehat{\varphi}_j$ converges in quadratic mean to $\dot{\varphi}_j$. The term $R_2$ is often known as an ``asymptotic bias'' or ``asymptotic drift'' term, and is typically second-order, i.e., can be bounded by a product or square of nuisance function errors.
Therefore, $R_2 = o_P(n^{-1/2})$ is implied when, for example,
relevant nuisance errors converge in $L_2$-norm to zero at rate $o_P(n^{-1/4})$. This dependence on the product of nuisance error biases has been referred to as rate double-robustness \citep{rotnitzky2021}. 
In sum, asymptotic linearity of the estimator $\widehat{\psi}_j$ is guaranteed when $\|\widehat{\varphi}_j-\dot{\varphi}_j\|=o_P(1)$, $R_2 = o_P(n^{-1/2})$, and when employing sample-splitting. 

The next result shows that the fraction of two asymptotically linear estimators is also asymptotically linear.

\begin{lemma}\label{lemma:asymptotic_linear}
Suppose for $j=1,2$, $\mathbb{P}_n\widehat{\varphi}_j-P_0\dot{\varphi_j} = \mathbb{P}_n\dot{\varphi}^*_j + o_P(n^{-1/2})$ where $\dot{\varphi}^*_j$ is a mean-zero function. Assuming there exists $\varepsilon >0$ such that $|\mathbb{P}_n\widehat{\varphi}_2| \wedge |P_0\dot{\varphi}_2| > \varepsilon$, 
    \begin{align*}
        \frac{\mathbb{P}_n\widehat{\varphi}_1}{\mathbb{P}_n\widehat{\varphi}_2}- \frac{P_0\dot{\varphi_1}}{P_0\dot{\varphi_2}} &= \mathbb{P}_n\left\{(P_0\dot{\varphi_2})^{-1}\left(\dot{\varphi}^*_1 -\dot{\varphi}^*_2\frac{P_0\dot{\varphi_1}}{P_0\dot{\varphi_2}}\right)\right\}+ o_P(n^{-1/2}).
    \end{align*}
\end{lemma}
\begin{proof}
    By adding and subtracting relevent terms, we obtain
    \begin{align*}
        \frac{\mathbb{P}_n\widehat{\varphi}_1}{\mathbb{P}_n\widehat{\varphi}_2}- \frac{P_0\dot{\varphi_1}}{P_0\dot{\varphi_2}} &= \frac{1}{\mathbb{P}_n\widehat{\varphi}_2}\left(\mathbb{P}_n\widehat{\varphi}_1-\mathbb{P}_n\widehat{\varphi}_2\frac{P_0\dot{\varphi_1}}{P_0\dot{\varphi_2}}\right) \\
        &=\frac{1}{P_0\dot{\varphi_2}}\left(\mathbb{P}_n\widehat{\varphi}_1-\mathbb{P}_n\widehat{\varphi}_2\frac{P_0\dot{\varphi_1}}{P_0\dot{\varphi_2}}\right)+\left(\frac{1}{\mathbb{P}_n\widehat{\varphi}_2}-\frac{1}{P_0\dot{\varphi_2}}\right)\left(\mathbb{P}_n\widehat{\varphi}_1-\mathbb{P}_n\widehat{\varphi}_2\frac{P_0\dot{\varphi_1}}{P_0\dot{\varphi_2}}\right).
    \end{align*}
    For the first term of the above display, it follows that
    \begin{align*}
        \mathbb{P}_n\widehat{\varphi}_1-\mathbb{P}_n\widehat{\varphi}_2\frac{P_0\dot{\varphi_1}}{P_0\dot{\varphi_2}}&=\left(\mathbb{P}_n\dot{\varphi}^*_1 +P_0\dot{\varphi_1}\right)-\left(\mathbb{P}_n\dot{\varphi}^*_2+P_0\dot{\varphi_2}\right) \frac{P_0\dot{\varphi_1}}{P_0\dot{\varphi_2}}+ o_P(n^{-1/2})\\
        &=\mathbb{P}_n\left(\dot{\varphi}^*_1 -\dot{\varphi}^*_2\frac{P_0\dot{\varphi_1}}{P_0\dot{\varphi_2}}\right)+ o_P(n^{-1/2}).
    \end{align*}
    For the second term, we obtain that 
    \begin{align*}
        &\left(\frac{1}{\mathbb{P}_n\widehat{\varphi}_2}-\frac{1}{P_0\dot{\varphi_2}}\right)\left(\mathbb{P}_n\widehat{\varphi}_1-\mathbb{P}_n\widehat{\varphi}_2\frac{P_0\dot{\varphi_1}}{P_0\dot{\varphi_2}}\right)\\
        &\qquad =\left(\frac{P_0\dot{\varphi_2}-\mathbb{P}_n\widehat{\varphi}_2}{\mathbb{P}_n\widehat{\varphi}_2P_0\dot{\varphi_2}}\right)\left((\mathbb{P}_n\widehat{\varphi}_1-P_0\dot{\varphi_1})+P_0\dot{\varphi_1}\left(1-\frac{\mathbb{P}_n\widehat{\varphi}_2}{P_0\dot{\varphi_2}}\right)\right).
    \end{align*}
    By the assumption that both $\mathbb{P}_n\widehat{\varphi}_2$ and $P_0\dot{\varphi_2}$ are bounded away from zero as well as the asymptotic linearity of both estimators, we conclude that
    \begin{align*}
        &\left(\frac{1}{\mathbb{P}_n\widehat{\varphi}_2}-\frac{1}{P_0\dot{\varphi_2}}\right)\left(\mathbb{P}_n\widehat{\varphi}_1-\mathbb{P}_n\widehat{\varphi}_2\frac{P_0\dot{\varphi_1}}{P_0\dot{\varphi_2}}\right) = O_P(n^{-1}) +  O_P(n^{-1}) = o_P(n^{-1/2}).
    \end{align*}
    This concludes the proof of the claim.
\end{proof}
The proof of Lemma 3.2 follows as a corollary. We introduced the following conditions in the main text:
\begin{enumerate}[label=\textbf{(A\arabic*)},leftmargin=2cm]
\setcounter{enumi}{4}
\item $\|\widehat{\pi}-\pi_0\| \max(\|\widehat{\lambda}-\lambda_0\|, \|\widehat{\mu}-\mu_0\|) = o_P(n^{-1/2})$
\label{as:doubl-robust}
\item $\|\widehat{\pi}-\pi_0\| =o_P(1)$, $\|\widehat{\mu}-\mu_0\| =o_P(1)$ and $\|\widehat{\lambda}-\lambda_0\| =o_P(1)$. \label{as:empirical-process}
\item $|Y|$ is bounded almost surely
\label{as:bounded-y} 
\end{enumerate}
\begin{proof}[\bfseries{Proof of Lemma 3.2}]
First assuming that sample-splitting is used to estimate $\mu_0$, $\lambda_0$ and $\pi_0$, Lemma 2 of \citet{kennedy2020sharp} implies that the empirical process term $R_1$ is controlled so long as $\lVert \widehat{\varphi}_j - \dot{\varphi}_j\rVert = o_P(1)$, for $j = 1, 2$, which is implied by ~\ref{as:empirical-process}, \ref{as:bounded-y} as well as the assumption that $\widehat \pi$ and $\pi$ are bounded away from zero almost surely (See Example 2 in Section 4.2 of \cite{kennedy2022semiparametric}). 

We will now establish $R_2 = o_P(n^{-1/2})$. By the well-known product bias for the influence function of the average treatment effect functional \citep{chernozhukov2018double}, and a standard application of the Cauchy-Schwarz inequality, we have 
\begin{align*}
    |P_0(\dot{\Gamma}_{\widehat{P}}-\dot{\Gamma}_{0})| \lesssim \|\widehat{\mu}-\mu_0\|\|\widehat{\pi}-\pi_0\| \quad \text{and} \quad |P_0(\dot{\Delta}_{\widehat{P}}-\dot{\Delta}_{0})| \lesssim \|\widehat{\lambda}-\lambda_0\|\|\widehat{\pi}-\pi_0\|,
\end{align*}
where we write ``$\|\widehat{\mu}-\mu_0\|$'' and ``$\|\widehat{\lambda}-\lambda_0\|$'' in place of $\|\widehat{\mu}(\cdot , 0) - \mu_0(\cdot, 0)\| + \|\widehat{\mu}(\cdot , 1) - \mu_0(\cdot, 1)\|$ and $\|\widehat{\lambda}(\cdot , 0) - \lambda_0(\cdot, 0)\| + \|\widehat{\lambda}(\cdot , 1) - \lambda_0(\cdot, 1)\|$, respectively (See Example 2 in Section 4.3 of \cite{kennedy2022semiparametric}).
. Therefore, $R_2$ is controlled if $\|\widehat{\mu}-\mu_0\|\|\widehat{\pi}-\pi_0\|=o_P(n^{-1/2})$ and $\|\widehat{\lambda}-\lambda_0\|\|\widehat{\pi}-\pi_0\|=o_P(n^{-1/2})$, which is implied by \ref{as:doubl-robust}. Next, we claim $|\mathbb{P}_n\widehat{\varphi}_2| \wedge |P_0\dot{\varphi}_2| > \varepsilon$ for some $\varepsilon > 0$. This is implied assuming $\Delta_0 > \varepsilon_1$ and $\varepsilon_3 < \widehat\pi(X)< 1-\varepsilon_3$. 
We now invoke Lemma~\ref{lemma:asymptotic_linear} where $\widehat{\varphi}_1$ and $\dot{\varphi_1}$ correspond to $\dot{\Gamma}_{\widehat{P}}$ and $\dot{\Gamma}_0$. Similarly, $\widehat{\varphi}_2$ and $\dot{\varphi_2}$ correspond to $\dot{\Delta}_{\widehat{P}}$ and $\dot{\Delta}_0$. Then Lemma~\ref{lemma:asymptotic_linear} states 
\begin{align*}
        \widehat\chi - \chi_0 &= \frac{\mathbb{P}_n\dot{\Gamma}_{\widehat{P}}}{\mathbb{P}_n\dot{\Delta}_{\widehat{P}}} - \frac{P_0\dot{\Gamma}_0}{P_0\dot{\Delta}_0} + o_P(n^{-1/2}) \\
        &= \mathbb{P}_n \left\{\frac{1}{\Delta_0}\left(\dot{\Gamma}_0^*- \dot{\Delta}_0^*\chi_0\right)\right\}+ o_P(n^{-1/2})\\
      &=\mathbb{P}_n\left\{\frac{1}{\Delta_0}\left(\frac{2Z - 1}
      {\pi_{0}} \left[Y - \mu_{0} -
        \chi_0\{A - \lambda_{0}\}\right] +
      \gamma_0 - \chi_0\delta_0\right)\right\} + o_P(n^{-1/2}).
\end{align*}
The first term of above display corresponds to the sample mean of the efficient influence function of $\chi_0$ as provided by Lemma 3.1 of the main text. Thus, 
 by multiplying both sides by $n^{1/2}$, it implies that 
\begin{align*}
        n^{1/2}(\widehat\chi - \chi_0) =n^{1/2}\mathbb{P}_n\dot{\chi}^*_0 + o_P(1).
\end{align*}
Additionally, it follows that 
\begin{align*}
    \text{Var}\left(n^{1/2}\mathbb{P}_n\dot{\chi}^*_0\right) = P_0\dot{\chi}^{*2}_0 < \infty
\end{align*}
assuming $\text{Var}[Y] < \infty$, $\Delta_0 > \varepsilon_1$ and $\varepsilon_2 < P_0(Z = 1 \mid X) < 1-\varepsilon_2$. Therefore, we conclude  
\begin{equation}
    n^{1/2}(\widehat\chi - \chi_0)\overset{d}{\longrightarrow} N\left(0, \,P_0\dot{\chi}^{*2}_0\right).\nonumber
\end{equation}
in view of the central limit theorem and the Slutsky's theorem
\end{proof}
\clearpage
\section{Profiling}
For the derivation of the identification results, we frequently refer to the following assumptions from the main text:
\begin{enumerate}[label=\textbf{(A\arabic*)},leftmargin=2cm]
\item \textbf{Relevance:} \label{as:relevance}$P_0\big(A(1) = A(0)\big) \neq 1$,
\item \label{as:unconfounded}\textbf{Effective random assignment:} For all $z,a \in \{0,1\}$, $Z \indep \big(A(z), Y(z, a)\big) \mid X$,
\item \label{as:excl_restriction}\textbf{Exclusion restriction:} $Y(z,a) = Y(z',a)$ for all $z,z',a \in \{0,1\}$, and
\item \label{as:monotonicity}\textbf{Monotonicity:} $P_0\big(A(1) < A(0)\big) = 0$.
\end{enumerate}
The following (well-known) results also become useful. 
\begin{lemma}\label{lm:id1}
    Assuming 
 SUTVA, \ref{as:unconfounded}, \ref{as:monotonicity}, and $0 < P_0(Z = 1 \mid X) < 1$ almost surely, then, for $V\subseteq X$, it follows that 
    \begin{align*}
        &P_0(A(1) > A(0)\mid V=v) = \E_0\big(A\mid Z=1, V=v\big)-\E_0\big(A\mid Z=0, V=v\big), \\
        &P_0(A(1) = A(0)= 1\mid V=v) = \E_0\big(A\mid Z=0, V=v\big),\quad \text{and}\\
        &P_0(A(1) = A(0)= 0\mid V=v) = \E_0\big(1-A\mid Z=1, V=v\big).
    \end{align*}
\end{lemma}
\begin{proof}
    Let $V'$ be $X \setminus V$ such that $X=(V, V')$. Furthermore, we assume that $V'$ follows the marginal density $d\widetilde{Q}$. It then follows that 
\begin{align*}
    P_0(A(1) > A(0)\mid V=v) &= \E_0(A(1) - A(0)\mid V=v) \\
    &= \int \E_0(A(1) - A(0)\mid X=x)\, d\widetilde{Q}(v')\\
    &= \int \left\{\E_0\big(A\mid Z=1, X\big)-\E_0\big(A\mid Z=0, X\big)\right\}\, d\widetilde{Q}(v') \\
    &= \E_0\big(A\mid Z=1, V=v\big)-\E_0\big(A\mid Z=0, V=v\big)
\end{align*}
where the first equality follows by \ref{as:monotonicity}, the second and the last equalities follow by the tower property of expectations, and the third equality follows by SUTVA, \ref{as:unconfounded}, and $0 < P_0(Z = 1 \mid X) < 1$ almost surely. Similarly, it follows that 
\begin{align*}
    &P_0(A(0) = A(1) = 1 \mid V=v)  \\
    &\qquad=\int \E_0(A(0)A(1) \mid X=x)\, d\widetilde{Q}(v')\\
    &\qquad= \int \E_0(A(0) \mid X=x)\, d\widetilde{Q}(v') && \text{By \ref{as:monotonicity}}\\
    &\qquad= \int \E_0(A(0) \mid X=x, Z=0)\, d\widetilde{Q}(v') &&\text{By \ref{as:unconfounded} and $P_0(Z=0\mid X)>0$ a.s.}  \\
    &\qquad= \int \E_0(A \mid X=x, Z=0)\, d\widetilde{Q}(v') &&\text{By SUTVA}\\
    &\qquad= \E_0(A \mid V=v, Z=0)
\end{align*}
and also that 
\begin{align*}
    &P_0(A(0) = A(1) = 0 \mid V=v)  \\
    &\qquad=\int \E_0(1-A(0)A(1) \mid X=x)\, d\widetilde{Q}(v')\\
    &\qquad= \int \E_0(1-A(1) \mid X=x)\, d\widetilde{Q}(v') && \text{By \ref{as:monotonicity}}\\
    &\qquad= \int \E_0(1-A(1) \mid X=x, Z=1)\, d\widetilde{Q}(v') &&\text{By \ref{as:unconfounded} and $P_0(Z=1\mid X)>0$ a.s.}  \\
    &\qquad= \int \E_0(1-A \mid X=x, Z=1)\, d\widetilde{Q}(v') &&\text{By SUTVA}\\
    &\qquad= \E_0(1-A \mid V=v, Z=1).
\end{align*}
This concludes the claim.
\end{proof}

Given the results above, the identification associated with the profiling parameters proceed as follows:
\begin{lemma}\label{lm:identification}Assuming 
 SUTVA, \ref{as:unconfounded}, \ref{as:monotonicity}, and $0 < P_0(Z = 1 \mid X) < 1$ almost surely, then, for $V\subseteq X$, it follows that 
 \begin{align*}
    &P_0(V=v_0\mid A(1) > A(0)) = \frac{\E_0[\mathds{1}(V=v_0) \{\E_0(A \mid X, Z=1)-\E_0(A \mid X, Z=0)\}]}{\E_0\{\E_0(A \mid X, Z=1)-\E_0(A \mid X, Z=0)\}}, \\
    &P_0(V=v_0\mid A(1) = A(0)= 1) = \frac{\E_0[\mathds{1}(V=v_0) \E_0(A \mid X, Z=0)]}{\E_0\{\E_0(A \mid X, Z=0)\}},\quad \text{and}\\
    &P_0(V=v_0\mid A(1) = A(0)= 0) =\frac{\E_0[\mathds{1}(V=v_0) \E_0(1-A \mid X, Z=1)]}{\E_0\{\E_0(1-A \mid X, Z=1)\}}.
\end{align*}
\end{lemma}
\begin{proof}
We denote by $Q$ the marginal distribution of $V$. Then, it follows that 
\begin{align}
P_0\left(V=v_0 \mid A(1)>A(0)\right) & =\int \mathds{1}(v=v_0) P\left(V=v\mid A(1)>A(0) \right) \, d Q(v) \nonumber\\
& = \frac{\int \mathds{1}(v=v_0)P_0\left(A(1)>A(0) \mid V=v\right)\, d Q(v) }{P_0\left(A(1)>A(0)\right)} \nonumber\\
& =\frac{\E_0[\mathds{1}(V=v_0) \{\E_0(A \mid X, Z=1)-\E_0(A \mid X, Z=0)\}]}{\E_0\{\E_0(A \mid X, Z=1)-\E_0(A \mid X, Z=0)\}}\nonumber
\end{align}
where the second equality follows by the Bayes rule and the last equality invokes the identification results provided by Lemma~\ref{lm:id1} under all required assumptions. Following the analogous steps, we also obtain
\begin{align}
P_0\left(V=v_0 \mid A(1)=A(0) = 1\right) & =\frac{\E_0[\mathds{1}(V=v_0) \E_0(A \mid X, Z=0)]}{\E_0\{\E_0(A \mid X, Z=0)\}}\nonumber
\end{align}
and 
\begin{align}
P_0\left(V=v_0 \mid A(1)=A(0) = 0\right) & =\frac{\E_0[\mathds{1}(V=v_0) \E_0(1-A \mid X, Z=1)]}{\E_0\{\E_0(1-A \mid X, Z=1)\}}\nonumber
\end{align}
in view of Lemma~\ref{lm:id1}.
\end{proof}

\subsection{Estimation and inference}
In this section, we discuss the estimation and inference of the parameter related to the profiling. Following the identification results given by Lemma~\ref{lm:identification}, we define the complier profile $\psi_{\text{co}}$
\begin{align*}
    \psi_{\text{co}}(P_0) := \frac{\E_0[\mathds{1}(V=v_0) \{\E_0(A \mid X, Z=1)-\E_0(A \mid X, Z=0)\}]}{\E_0\{\E_0(A \mid X, Z=1)-\E_0(A \mid X, Z=0)\}} = \frac{\E_0[\mathds{1}(V=v_0) \delta_0(X)]}{\E_0[\delta_0(X)]}
\end{align*}
As discussed in the main text, the denominator of the above display can be estimated by $\mathbb{P}_n\dot{\Delta}_{\widehat P}$. We thus  focus on the estimation of the numerator. We first derive the influence function of the functional $\E_0[\mathds{1}(V=v_0) \E_0(1-A \mid X, Z=1)]$. Following \cite{kennedy2022semiparametric}, we denote by $\mathbb{IF}$ an operator that maps real-valued functionals to their (efficient) influence functions in a nonparametric model. The influence function of the numerator can be obtained as follows:
\begin{align*}
&\mathbb{IF}\left(\E_P\big[\delta_P(X)\mathds{1}(V=v_0)\big]\right) \\
&\qquad= \sum_{x}\sum_{v}\mathbb{IF}\{\delta_P(x)p(x)\}\mathds{1}(v=v_0)+\delta_P(x)p(x)\mathbb{IF}\{\mathds{1}(v=v_0)\}\\
    &\qquad= \sum_{x}\sum_{v}\mathbb{IF}\{\delta_P(x)p(x)\}\mathds{1}(V=v_0)+\delta_P(x)p(x)\{\mathds{1}(v=v_0)-P(v=v_0)\}\\
    &\qquad= \sum_{x}\sum_{v}\left(\mathbb{IF}\{\delta_P(x)p(x)\}+\delta_P(x)p(x)\right)\mathds{1}(v=v_0)-\delta_P(x)p(x)P(v=v_0)\\
    &\qquad= \dot{\Delta}_P(a, z, x)\mathds{1}(v=v_0)-\E_P\big[\delta_P(X)\mathds{1}(V=v_0)\big]
\end{align*}
where the first line is the application of the product rule of the influence function \citep{kennedy2022semiparametric} and the last line is derived by \cite{kennedy2022semiparametric} as Example 2. This result relies on the assumption that $V$ is a discrete random variable. When $V$ is continuous, the influence function of $v \mapsto \mathds{1}(v=v_0)$ does not exists since it fails to satisfy the required differentiability assumption. We now apply Lemma~\ref{lemma:asymptotic_linear} in this context and obtain the following result:
   \begin{align*}
        &\frac{\mathbb{P}_n\dot{\Delta}_{\widehat P}\mathds{1}(v=v_0)}{\mathbb{P}_n\dot{\Delta}_{\widehat P}}- \psi_{\text{co}} \\
        &\quad = \mathbb{P}_n\left\{\frac{1}{\Delta_0}\left(\dot{\Delta}^*_0 -(\dot{\Delta}_0\mathds{1}(V=v_0)-P_0\delta_0\mathds{1}(V=v_0))\psi_{\text{co}}\right)\right\}+ o_P(n^{-1/2}).
    \end{align*}
assuming (i) $\|\widehat{\pi}-\pi_0\|\|\widehat\lambda -\lambda_0\| = o_P(n^{-1/2})$ and (ii) there exist $\varepsilon_1, \varepsilon_2, \varepsilon_3 > 0$ such that $\Delta_0 > \varepsilon_1$, $\varepsilon_2 < P_0(Z = 1 \mid X) < 1-\varepsilon_2$ and $\varepsilon_3 < \widehat\pi(X)< 1-\varepsilon_3$ almost surely. Furthermore, we denote the influence function of $\psi_{\text{co}}$ at $P$ as follows:
\begin{equation}
    \dot\varphi^*_{P, \text{co}} := (a, z, x) \mapsto \frac{1}{\Delta_P}\left(\dot{\Delta}^*_P -(\dot{\Delta}_P\mathds{1}(V=v_0)-P\delta_P\mathds{1}(V=v_0))\psi_{\text{co}}(P)\right)
\end{equation}
Then the $(1-\alpha)$-level confidence interval can then be obtained by
\begin{equation}
    \frac{\mathbb{P}_n\dot{\Delta}_{\widehat P}\mathds{1}(V=v_0)}{\mathbb{P}_n\dot{\Delta}_{\widehat P}} \pm n^{-1/2}q_{1-\alpha/2} \left(\mathbb{P}_n\dot\varphi^{*2}_{\widehat P, \text{co}}\right)^{1/2}
\end{equation}
where $\mathbb{P}_n \dot\varphi^{*2}_{\widehat P, \text{co}}$ is an estimated asymptotic variance of the estimator. 

Following the identification results given by Lemma~\ref{lm:identification}, we define the always-taker profile $\psi_{\text{at}}$
\begin{align*}
    \psi_{\text{at}}(P_0) := \frac{\E_0[\mathds{1}(V=v_0) \E_0(A \mid X, Z=0)]}{\E_0\{\E_0(A \mid X, Z=1)\}} = \frac{\E_0[\mathds{1}(V=v_0) \lambda_0(X, 0)]}{\E_0[\lambda_0(X, 0)]}
\end{align*}
Example 2 of \cite{kennedy2022semiparametric} provides the influence function of the denomenator above, which is given by
\begin{align*}
    \dot\Lambda^*_{P,0} &:= (a, z, x) \mapsto \frac{1-z}{1-\pi_P(x)}\left\{a-\lambda_P(x, 0)\right\} + \lambda_P(x, 0) - \E_P[ \lambda_P(X, 0)]\quad \text{and}\\
    \dot\Lambda_{P,0} &:= (a, z, x) \mapsto \frac{1-z}{1-\pi_P(x)}\left\{a-\lambda_P(x, 0)\right\} + \lambda_P(x, 0) 
\end{align*}
Then by an analogous derivation for the complier profile, the influence function of $\psi_{\text{at}}$ can be obtained as 
\begin{align*}
    \dot\varphi^*_{P, \text{at}}  (a, z, x) &:= \frac{1}{\E_P[ \lambda_P(X, 0)]}\big[\dot\Lambda^*_{P,0}(a, z, x)  \\
    &\qquad -\{\dot\Lambda_{P,0}(a, z, x) \mathds{1}(v=v_0)-\E_P[ \lambda_P(X, 0)\mathds{1}(v=v_0)]\}\psi_{\text{at}}(P)\big]
\end{align*}
where $\dot\lambda$ is the uncentered influence function. 
Then the $(1-\alpha)$-level confidence interval can then be obtained by
\begin{equation}
    \frac{\mathbb{P}_n\dot{\Lambda}_{\widehat P, 0}\mathds{1}(V=v_0)}{\mathbb{P}_n\dot{\Lambda}_{\widehat P, 0}} \pm n^{-1/2}q_{1-\alpha/2} \left(\mathbb{P}_n\dot\varphi^{*2}_{\widehat P, \text{at}}\right)^{1/2}
\end{equation}
Following the identification results given by Lemma~\ref{lm:identification}, we define the never-taker profile $\psi_{\text{at}}$
\begin{align*}
    \psi_{\text{nt}}(P_0) := \frac{\E_0[\mathds{1}(V=v_0) \E_0(1-A \mid X, Z=1)]}{\E_0\{\E_0(1-A \mid X, Z=1)\}} = \frac{\E_0[\mathds{1}(V=v_0) \{1-\lambda_0(X, 1)\}]}{\E_0[1-\lambda_0(X, 1)]}
\end{align*}
Example 2 of \cite{kennedy2022semiparametric} provides the influence function of the denomenator above and its uncentered analogue, which are given by
\begin{align*}
    \dot\Lambda^*_{P,1} &:= (a, z, x) \mapsto \frac{z}{\pi_P(x)}\left\{\lambda_P(x, 1)-a\right\} + 1-\lambda_P(x, 1) - \E_P[1- \lambda_P(X, 1)]\quad \text{and}\\
    \dot\Lambda_{P,1} &:= (a, z, x) \mapsto \frac{z}{\pi_P(x)}\left\{\lambda_P(x, 1)-a\right\} + 1-\lambda_P(x, 1).
\end{align*}
Then by an analogous derivation for the complier profile, the influence function of $\psi_{\text{nt}}$ can be obtained as 
\begin{align*}
    \dot\varphi^*_{P, \text{nt}}(a, z, x) & := \frac{1}{\E_P[ 1-\lambda_P(X, 1)]}\big[\dot\Lambda^*_{P,1}(a, z, x) \\
   &\qquad-\{\dot\Lambda_{P,1}(a, z, x)\mathds{1}(v=v_0)-\E_P[ (1-\lambda_P(X, 1))\mathds{1}(v=v_0)]\}\psi_{\text{nt}}(P)\big]
\end{align*}
where $\dot\lambda$ is the uncentered influence function. 
Then the $(1-\alpha)$-level confidence interval can then be obtained by
\begin{equation}
    \frac{\mathbb{P}_n\dot{\Lambda}_{\widehat P, 1}\mathds{1}(V=v_0)}{\mathbb{P}_n\dot{\Lambda}_{\widehat P, 1}} \pm n^{-1/2}q_{1-\alpha/2} \left(\mathbb{P}_n\dot\varphi^{*2}_{\widehat P, \text{nt}}\right)^{1/2}
\end{equation}
\clearpage

\clearpage

\bibliographystyle{asa}
\bibliography{learner}


\noindent%
	{\bfseries\large
Supplement B: Additional Simulation Results}

\noindent%
	
	\noindent%

\pagenumbering{arabic}
\setcounter{equation}{0}
\setcounter{table}{0}
\setcounter{figure}{0}
\renewcommand{\theequation}{S\arabic{equation}}
\renewcommand{\thetable}{S\arabic{table}}
\renewcommand{\thefigure}{S\arabic{figure}}
\renewcommand\theassumption{S.\arabic{assumption}}
\renewcommand\thetheorem{S.\arabic{theorem}}

\section{Scenario 2}
Under the second scenario, we introduce a relatively minor form of nonlinearity. Specifically, we now use the following specification for the key models:
\begin{itemize}
\item
  $\pi_0(X, 1) = \mathrm{expit}(0.4 X_1 - 0.8X_2)$
\item
  $\lambda_0^{\dagger}(X, Z, U) = \mathrm{expit}(-0.3 - 0.4X_1 - 0.14 X_2 + 1.1Z + 0.7 U)$
\item
  $r(\boldsymbol{X}) = - 4 X_1 + 6 \{X_2 - 0.3\} $
\item $s(\boldsymbol{X}) = 40 - 7X_1 - 8X_2$
\end{itemize}
Under this DGP, all the main models are based on linear functions. The only form of nonlinearity is due to the logistic links in $\pi_0$ and $\lambda_0^{\dagger}$. As such, we might expect TSLS and the parametric DRML IV estimator to display considerably less bias due to model misspecification compared to scenario 1. Normally, we might expect parametric methods to outperform nonparametric methods due to the bias-variance tradeoff. Here, we investigate whether the nonparametric IV estimator loses efficiency compared to parametric methods.

Figure~\ref{fig:scenario2} contains the results from the simulation under scenario 2. Under this scenario, we find that both DRML methods provide nearly identical results. There is a small amount of bias for the two smallest sample sizes, and bias shrinks to essentially zero for larger sample sizes. Notably, TSLS is biased, and that bias does not shrink with larger sample sizes. When we review the results for RMSE, we find that there is no difference between the two DRML methods, and only slight differences between the DRML methods and TSLS. These results demonstrate the strength of DRML methods. More typically, we might expect slow rates of convergence for a flexible estimator, but we observe that DRML estimation methods allow for flexible fits with parametric-like efficiency. 

\begin{figure}[htbp]
  \caption{Comparing bias and RMSE for TSLS and DRML IV under simulation scenario 2. }
  \label{fig:scenario2}
  \centering
  \includegraphics[width=0.9\linewidth]{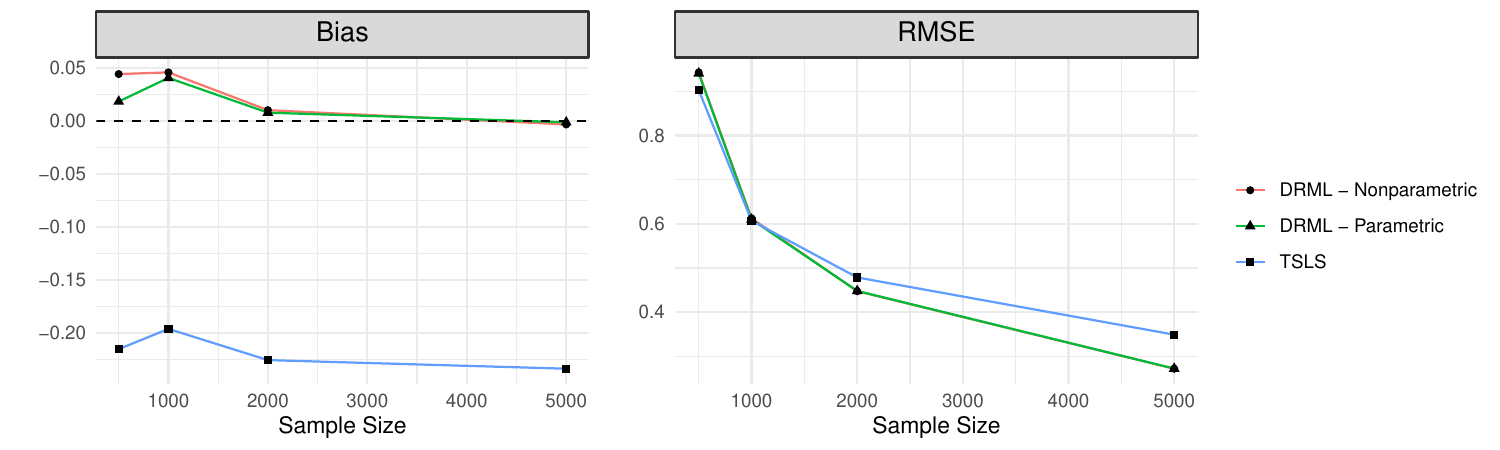}
\end{figure}

\section{Scenario 3}

Under the third scenario, we use functional forms for the models such that the model misspecification falls somewhere between scenarios 1 and 2. Now we only include nonlinearity for the model $\pi_0$ and $\lambda_0^{\dagger}$. Note that the form of nonlinearity in these two models is identical to that in scenario 1. As such, the DGP for scenario 3 is as follows: 

\begin{itemize}
\item
  $\pi_0(X, 1) = \mathrm{expit}(0.4 X_1 - 0.8X_2 + 0.4\mathds{1}(X_1 > 0))$
\item
  $\lambda_0^{\dagger}(X, Z, U) = \mathrm{expit}(-0.3 - 0.4X_1 - 0.14 X_2
  + 1.1 Z - 0.55 X_1 Z- 0.7 \mathds{1}(X_1 > 0) + 0.7 U)$
\item
  $r(\boldsymbol{X}) = - 4 X_1 + 6 \{X_2 - 0.3\}$
\item $s(\boldsymbol{X}) = 40 - 7X_1 - 8X_2$
\end{itemize}

Figure~\ref{fig:scenario3} contains the results for the simulation from scenario 3. The general pattern here mirrors that of scenario 2.  That is, both DRML IV estimators display less bias compared to TSLS. In terms of RMSE, all three methods have nearly identical behavior. In general, our simulation study makes a strong case for DRML IV estimation methods. That is, under various forms of model  misspecification TSLS is biased by differing amounts. DRML IV methods produce estimates with smaller bias than TSLS in all three scenarios. However, under more extreme forms of model misspecification, the nonparametric DRML estimator performed the best. Critically, despite being based on flexible nonparametric estimators, its performance is never worse in terms of RMSE. As such, this allows analysts to employ flexible estimation methods without sacrificing efficiency even when parametric assumptions hold.
 
\begin{figure}[htbp]
  \caption{Comparing bias and RMSE for TSLS and DRML IV under simulation scenario 3. }
  \label{fig:scenario3}
  \centering
  \includegraphics[width=0.9\linewidth]{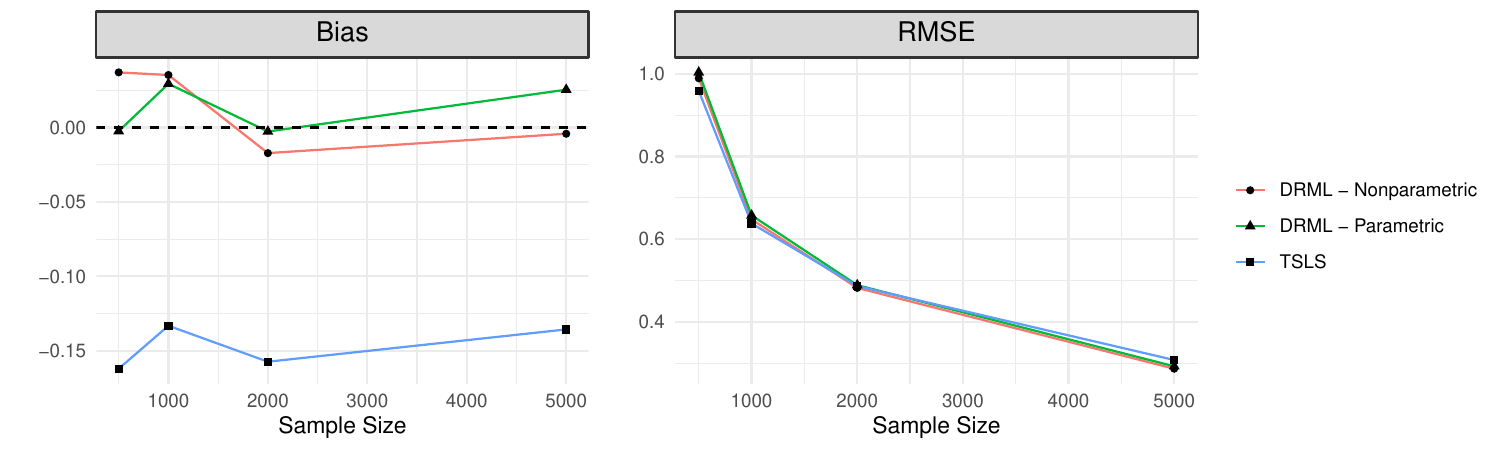}
\end{figure}
\clearpage
\bibliographystyle{asa}
\bibliography{learner}